\documentclass[runningheads]{llncs}
\usepackage[T1]{fontenc}
\usepackage{graphicx}
\usepackage{macro}
\usepackage{amssymb}
\usepackage{amsmath}
\usepackage[table]{xcolor}
\usepackage{hyperref}
\usepackage{cleveref}
\usepackage[ruled,linesnumbered]{algorithm2e}
\usepackage[final]{microtype}

\crefname{algocf}{Algorithm}{Algorithms}
\Crefname{algocf}{Algorithm}{Algorithms}

\usepackage{color}

\AtBeginEnvironment{proof}{\setlength{\parskip}{0pt}}

\begin{document}
\setlength{\textfloatsep}{12pt plus 2pt minus 2pt}
\setlength{\floatsep}{12pt plus 2pt minus 2pt}
\setlength{\intextsep}{12pt plus 2pt minus 2pt}

\allowdisplaybreaks

\title{On the Power of Deception in Repeated Games}
\author{
Saba Ahmadi \and 
Avrim Blum \and 
Dimitar Chakarov \and
Melissa Dutz 
 %\inst{1}
}

\authorrunning{Ahmadi et al.}
\institute{Toyota Technological Institute at Chicago \\
\email{\{saba, avrim, chakarov, melissa\}@ttic.edu}
}
\maketitle              % typeset the header of the contribution
\begin{abstract}
  In repeated games, opponents often predict what we'll do next by looking at what we have done so far. This allows us to deceive them: we can deliberately behave one way for a period of time to shape their expectations, then switch strategies to profit from the induced response. We study deception in repeated two-player normal-form games against \textit{count-based} learners, whose behavior depends only on how often we have played each action in the past. We formalize deceptive and non-deceptive play, and introduce the notion of a \emph{deception bonus}, the payoff gain of the best deceptive strategy over the best fixed mixed strategy. We establish structural results on deception in general-sum games. We design exact dynamic programs for optimizing against any count-based learner when the action space or opponent's memory is small, and develop approximation algorithms for settings where the opponent's memory or the time horizon is large. We also provide an approximation algorithm for learning to deceive an opponent whose count-based learning rule is unknown. To complement our algorithmic results, we show that approximating the optimal deceptive payoff against the classic Empirical Risk Minimization (ERM) learning rule is NP-hard, including obtaining any constant-factor approximation or even a $T^\alpha$-additive approximation for any $0 < \alpha < 1$. Finally, we empirically measure the deception bonus in random games with i.i.d. payoffs.

\end{abstract}

% Paper body
\section{Introduction}
\label{section:introduction}

In this paper, we propose and explore a game-theoretic model of {\em deception}, and consider the power of {deceptive} play in repeated, normal-form, two-player games.
What constitutes deceptive play? To deceive someone is to cause them to believe something which is untrue, so to study \textit{deceiving} an opponent requires some model of what the opponent ``believes.'' Here, we model the opponent as believing our past play is predictive of our future play: we study deceiving ``count-based'' learners, whose behavior is a function of the empirical counts (i.e.\ empirical
distribution) of our past actions. We have a special focus on opponents who use the classic Empirical Risk Minimization algorithm (and a limited-memory variant). Empirical Risk Minimization ($\FTL$) is widely used in settings where one believes the world is stochastic. Using $\FTL$, the opponent plays the action which would have achieved the highest cumulative payoff against our past actions. This models an opponent who believes our next action will be drawn from the empirical distribution of our historical actions ($\FTL$ is an optimal strategy in that case). $\FTL$ is also known as ``Follow-the-Leader'' in strategic games; we refer to it as $\FTL$ here to reflect the opponent's incorrect belief in the stochasticity of the interaction. A deceptive player could take advantage of $\FTL$ by playing a strategic sequence of actions to effectively lure their opponent into playing a favorable action, in order to then deviate from the opponent's expectations and reap a big payoff. In practice, for example, this could capture the behavior of a “pool shark”: someone who pretends to be unskilled at pool by playing poorly in order to lure an opponent into playing for money, then surprises their opponent by playing well and winning the money. 

In this setting, we ask, how can the deceptive player---who we refer to as ``the optimizer''---achieve the highest possible payoff? How much does the optimizer stand to gain by deception over what can be gained by ``non-deceptive'' play? This raises the question of what, in general, constitutes ``non-deceptive'' play.

Announcing which action one will play next would certainly be transparent. But this clearly leaves one open to exploitation; in many zero-sum games, for example, the only minimax optimal strategies are mixed. What about announcing a mixed strategy? Randomization does serve to ``hide'' what one will do; is this deceptive? Previous works often model deception as randomization; for example, randomization is used to model bluffing in the Kuhn poker model~\cite{Kuhn+1951+97+104}. In this work, however, we will consider using randomization to be simply a form of unpredictability, and not deceptive in and of itself. We define our non-deceptive benchmark to be the maximum payoff one could achieve by announcing and playing according to a fixed mixed strategy.  In particular, playing a minimax-optimal strategy is {\em not} deceptive in our model.

\subsection{Our Contributions} 
In the real world, deception often unfolds over time, initially building up trust or shaping expectations to lure someone into a vulnerable position, and then taking advantage. We introduce a new model of deception which aims to capture this kind of multi-stage deception. By contrast, existing works tend to model deception as a one-shot use of randomization (as mentioned above, bluffing in Kuhn poker is an example of this).

We also introduce a new notion of non-deceptive play and a ``deception bonus,'' which measures the additional payoff deceptive play could garner over non-deceptive play.

We begin by establishing structural results for deception in general-sum games. We then design dynamic programming algorithms for exactly optimizing against any count-based learner when the game or the opponent's memory is small, as well as approximation algorithms for memory-limited count-based learners when the opponent's memory or the number of rounds is large. We also provide an algorithm for learning to approximately optimize against an opponent whose specific count-based learning rule is unknown. Complementing our algorithmic results, we show that approximating the optimal payoff against $\FTL$ even in zero-sum games is NP-hard, improving on results of \cite{assosmaximizing}. Finally, we measure the deception bonus in random zero- and general-sum games.

\subsection{Related work}
\paragraph{Optimizing against Bounded Players.} 
\cite{fortnow1994optimality} studied playing against computationally bounded players whose strategies must be realizable by finite automata or polynomial-time Turing machines. They showed that for the games Prisoner's Dilemma and Matching Pennies, there exist dominant strategies for playing against a polynomial-time opponent. \cite{freund1995efficient} extended this line of work by giving a polynomial-time algorithm to learn approximately optimal play against any finite automaton with probabilistic state transitions. Building on that, they further established that if transitions of the (probabilistic state) automaton can be randomized, finding an approximately optimal strategy is PSPACE-complete.

\paragraph{Optimizing against Online Learners.}
Recent work has established computational hardness for optimizing against opponents who use specific online learning algorithms. \cite{assosmaximizing} show that optimizing against $\FTL$ is computationally intractable in general-sum games. In contrast, our hardness result from \Cref{section:hardness-results} shows that the problem is hard even for zero-sum games, and further establishes hardness of approximation up to additive $T^{\alpha}$ for $0 < \alpha < 1$. A complementary line of work~\cite{braverman2018selling,deng2019strategizing,mansour2022strategizing,brown2023learning} studies how to optimize against no-regret online learners. In particular, \cite{deng2019strategizing} show that in general-sum games, the optimizer can gain higher cumulative utility when strategizing against a no-regret learner compared to the Stackelberg equilibrium utility. In other words, \cite{deng2019strategizing} show that a deception bonus exists in games where the opponent has more than two actions and is playing a certain kind of no-regret strategy which limits the probability mass on actions that would have performed sufficiently sub-optimally in hindsight (also known as a ``mean-based'' strategy). In contrast, \cite{mansour2022strategizing} characterize the class of learning algorithms that limit the optimizer's payoff to at most the Stackelberg value. That is, they study the class of learning rules that permit no deception bonus.

\paragraph{Other Notions of Deception.} 
Previous works study deception in Stackelberg games, where the follower misreports their type~\cite{gan2019imitative} or commits to a fake payoff matrix~\cite{nguyen2019imitative,nguyen2020tackling} in order to elicit a specific response from the leader. Later works explore how the follower can optimally deceive a learning leader~\cite{birmpas2020optimally} or learn a good deceptive payoff matrix~\cite{chen2023learning}. In hidden-role games~\cite{aitchison2021learning,carminati2024hidden}, where players are assigned private team roles and have to recognize and cooperate with their own team members, deception arises in choosing what information to reveal to the other players. In contrast to these notions where deception involves commitment to false information or concealment of private information, our work models deception over time by capturing how a player can strategically choose a sequence of actions to manipulate and exploit the opponent's learning algorithm. In a different direction, \cite{balcan2018diversified} investigated diversified strategies that protect against deceptive play by bounding the maximal probability of a single action.
\section{Model}\label{section:model}

We consider a repeated normal-form game with two players---$A$ (the optimizer) and $B$ (the opponent)---played for $T$ rounds. The optimizer $A$ has $k_A$ actions $\calA = \{a_1, \dots, a_{k_A}\}$ and the opponent $B$ has $k_B$ actions $\calB = \{b_1, \dots, b_{k_B}\}$. Let $u^A, u^B : \calA \times \calB \to \mathbb{R}$ be the utility functions of $A$ and $B$, respectively. Notice that utility is a function of a single pair of actions, i.e. each player receives utility once per round. 

\begin{definition}[Expected payoff for player $A$]
\label{def:expected-payoff}
    For $a \in \calA$ and $\beta \in \Delta(\calB)$ (where $\Delta(\calB)$ denotes the set of mixed actions available to player $B$), the \emph{expected payoff} for player~$A$ is
    \[
        \bar{u}^A(a,\beta) \triangleq \E_{b \sim \beta}\!\left[u^A(a,b)\right] = \sum_{j=1}^{k_B} \beta(b_j) u^A(a,b_j).
    \]
\end{definition}

\subsection{Opponent Strategies}

Throughout the paper, we focus on deceiving opponents who use $\emph{count-based}$ learning rules.

\begin{definition}[Count-based learning rule]
\label{definition:count-based}
A count-based learning rule for the opponent, player $B$, is a function $f: \mathbb{N}^{k_A} \to \Delta (\mathcal{B})$ that maps counts $c_1 \dots c_{k_A}$, where $c_i$ is the number of times $A$ has played action $a_i$, to a mixed strategy over the action space $\mathcal{B}$. When $f$ outputs pure strategies, we write $f : \mathbb{N}^{k_A} \to \mathcal{B}$. At the beginning of round $t$, if the opponent has limited memory length $m$, $c_i$ is the number of times $A$ played $a_i$ in the last $m$ rounds, and $\sum_{i \in [k_A]}c_i = \min(m, t-1)$. Otherwise, $\sum_{i \in [k_A]}c_i = t-1$.
\end{definition}

While a count-based learning rule does not depend on the time-order of actions in the sense that the mixed action it plays is a function only of the action counts, by convention, we also write $f([a_1, \dots, a_t])$ to denote the mixed strategy obtained by applying $f$ to the counts induced by the action sequence $[a_1, \dots, a_t]$. For limited-memory count-based learning rules, access to the sequence is used exclusively to maintain accurate action counts over a sliding window of the past $m$ rounds: at each time step, the opponent must know which action was played $m$ steps ago in order to know which count to decrement next. The strategy itself remains a function only of the action counts.

\begin{remark}
 For a given count-based learning rule, we assume in our algorithmic results that we can compute its output given historical counts in $O(k_A k_B)$ time. This allows constant-time computations involving all $k_A$ counts for each of the opponent's $k_B$ actions. If a given count-based algorithm can be computed faster or requires more time, this factor would change accordingly in our runtimes.  
\end{remark}

For our structural and hardness results, we focus on the Empirical Risk Minimization (\FTL) learning rule and its memory-limited variant, which are both contained in the class of count-based learning rules.

\begin{definition}[Empirical Risk Minimization, $\FTL$]\label{definition:ftl}
    At time step $t$, the learning rule \emph{Empirical Risk Minimization} $\FTL : \calA^{t - 1} \to \calB$ outputs the action in $\calB$ that achieves the best cumulative utility in hindsight, so
    \[
        \FTL([a_1, \dots, a_{t - 1}]) = \argmax_{b \in \mathcal{B}} \sum_{i = 1}^{t - 1} u^B(a_{i}, b),
    \]
    with ties broken deterministically according to a fixed order (e.g. lexicographically). Because $\FTL$ does not use the time-order of the actions in its input, we also use $\FTL(c_1, \dots, c_{k_A}) = \argmax_{b \in \mathcal{B}} \sum_{i = 1}^{k_A} c_i \cdot u^B(a_{i}, b),$ where $c_i$ is the number of times action $a_i \in \calA$ was used by player $A$ in the last $t = \sum_{i = 1}^{k_A} c_i$ steps. 
\end{definition}

\begin{definition}[Limited-memory Empirical Risk Minimization, $\FTLm$]\label{definition:ftl-m}
    At time step $t$, the learning rule \emph{Empirical Risk Minimization with memory $m$} $\FTLm : \calA^{\min(m, t - 1)} \to \calB$ outputs the action in $\calB$ that achieves the best cumulative utility over the last $m$ actions of $A$, so
    \[
        \FTLm([a_{\max(1, t - m)}, \dots, a_{t - 1}]) = \argmax_{b \in \mathcal{B}} \sum_{i = \max(1, t - m)}^{t - 1} u^B(a_{i}, b).
    \]
    Similarly to~\Cref{definition:ftl}, we use 
    \[
        \FTLm(c_1, \dots, c_{k_A}) = \argmax_{b \in \mathcal{B}} \sum_{i = 1}^{k_A} c_i u^B(a_{i}, b)
    \]
    for $\sum_{i = 1}^{k_A} c_i \leq m$.
\end{definition}

When discussing $\FTL$ and $\FTLm$, it will be helpful to refer to the \textit{score} of an action, defined below. 

\begin{definition}[Score of an action or column]\label{definition:score}
    The \emph{score of an action} $b \in \calB$ for player $B$ given that player $A$ has played the sequence $a_1, \dots, a_t$ is $\sum_{i = 1}^t u^B(a_i, b)$, where $t \leq m $ when the opponent has memory limit $m$. Using this definition, an opponent playing $\FTL$ or $\FTLm$ selects the action which has the maximum score. When working with a payoff matrix, since player $B$ is the column player, we also refer to the score of an action as the score of a column.
\end{definition} 

\subsection{Deception and the Non-Deceptive Benchmark}

\begin{definition}[Non-deceptive play]
    Playing according to a fixed, publicly-known mixed strategy. 
\end{definition}

By this definition of non-deceptive play, the best non-deceptive strategy would be the Stackelberg leader strategy, the utility-maximizing mixed strategy against the best response of player $B$. We refer to any additional payoff a player could achieve against a particular opponent in a given game compared to the best non-deceptive strategy (i.e., a Stackelberg leader strategy) as the \textit{deception bonus}.

\begin{definition}[Deception bonus]\label{definition:deception_bonus}
    Let $\bfa$ be player $A$'s Stackelberg leader strategy, i.e. a utility-maximizing mixed strategy against the best response of player $B$. Let $\bfb = \argmax_{b \in \calB} \E_{a' \sim \bfa}\!\left[ u^B(a', b)\right]$ be the best response of player $B$. Then the \emph{deception bonus} of using an optimal sequence of actions as opposed to $\bfa$ is
    \begin{equation*}
        \left[\max_{\{a_t\}_{t = 1}^T} \sum_{t = 1}^T \bar{u}^A(a_t, b_t)\right] - 
        T \cdot E_{a'\sim\bfa}\!\left[ u^A(a', \bfb)\right],
    \end{equation*}
    where $b_t$ is the mixed action of player $B$ at time $t$.
\end{definition}

 In \Cref{definition:deception_bonus}, if player $B$ implements $\FTL$, then $b_t = \FTL([a_1, \dots, a_{t - 1}])$, and if player $B$ implements $\FTLm$, then $b_t = \FTLm([a_{\max(1, t - m)}, \dots, a_{t - 1}])$. Note that in the limit as $T \rightarrow \infty$, for any game whose payoffs are fixed (i.e.\ not dependent on $T$), if player $A$ played a fixed mixed strategy without announcing it in advance, $\FTL$ would be best responding.

\section{Structural Properties of Deception}
\label{sec:structural_observations}

\newtheorem{observation}{Observation}
\newcommand{\obs}[2]{
  \begin{observation}
  \textit{#1}\par
  \normalfont #2
  \end{observation}
}

We begin with structural results about the existence of deception and size of the deception bonus. In this section, we consider exclusively opponents who use \FTL as their learning algorithm. The rows of each game matrix correspond to the optimizer's actions, and columns correspond to the opponent's actions. We assume the opponent breaks ties between actions (columns) by choosing the leftmost tied column. Each game is played for $T$ rounds.

\begin{table}[htbp]
\centering
\caption{Game with arbitrarily large deception bonus.}
\begin{tabular}{c|c|c|c|}
    & $a$ & $b$ & $c$ \\
    \hline
    $a$ & $(-1, 1)$ & $(1, -1)$ & $(0, \frac{\epsilon}{T})$ \\
    \hline
    $b$ & $(1, -1)$ & $(-1, 1)$ & $(0, \frac{\epsilon}{T})$ \\
    \hline
    $c$ & $(0, 0)$ & $(0, 0)$ & $(B, -B)$ \\
    \hline
\end{tabular}
\label{tab:large_bonus_infinite}
\end{table}

\begin{proposition}
\label{lem:bonus-is-large}
    The deception bonus can be arbitrarily large.
\end{proposition}
\begin{proof}
    Consider the game shown in \Cref{tab:large_bonus_infinite}. If $a$ appears more than $b$ in the optimizer's history, the opponent will play $a$, and vice versa. If the history includes only $a$ and $b$ an equal number of times, the opponent will play $c$, and the optimizer can also play $c$ to get a payoff of $B$, where $B$ is some arbitrarily large number. Playing $b, a, c$ followed by best responding in the remaining $T-3$ rounds, the optimizer (row player) achieves a utility of at least $(T - 1) + B$. On the other hand, the Stackelberg leader strategy is to put equal probability on $a$ and $b$, and a small enough probability $p_c$ on $c$ such that $c$ is a best response for the column player; in particular, $p_c=\frac{\epsilon}{\epsilon + BT}$. The expected utility of the Stackelberg leader strategy over $T$ rounds is $p_c\cdot B \cdot T = \frac{\epsilon \cdot B \cdot T}{\epsilon + BT}$ which approaches $\epsilon$ as $B, T \rightarrow \infty$. 
\end{proof}

\begin{table}[htbp]
\centering
\caption{
Game where the deception bonus is $\bigomega{T}$, even when payoffs are bounded by a constant.
}
\begin{tabular}{c|c|c|c|}
    & $a$ & $b$ & $c$ \\
    \hline
    $a$ & $(-1, 1)$ & $(1, -1)$ & $(0, \frac{\epsilon}{T})$ \\
    \hline
    $b$ & $(1, -1)$ & $(-1, 1)$ & $(0, \frac{\epsilon}{T})$ \\
    \hline
\end{tabular}
\label{tab:large_bonus}
\end{table}

\begin{proposition}
\label{lem:bonus-is-linear}
    The deception bonus can be $\bigomega{T}$ even when payoffs are bounded by a constant.
\end{proposition}
\begin{proof}
    Consider the game shown in \Cref{tab:large_bonus}. In this game, the Stackelberg leader strategy for the row player is to play $a$ and $b$ with probability 1/2 each, resulting in an expected overall payoff of 0. However, if the optimizer (row player) plays $baba\dots$, then $\FTL$ will play $abcbc\dots$, with ${bc}$ repeating indefinitely, resulting in $\bigomega{T}$ utility for the optimizer. 
\end{proof}

\begin{table}[htbp]
\caption{Game with no deception bonus.}
\label{tab:no_bonus}
\centering
\(
\begin{array}{c|c|c|c|}
     & a & b \\
    \hline
    a & (0, 1) & (10, 0) \\
    \hline
    b & (1, 1) & (0, 0) \\
    \hline
\end{array}
\)
\end{table}

\begin{proposition}
\label{lem:no-deception-bonus}
    There exist games with no deception bonus.
\end{proposition}
\begin{proof}
    Consider the game in~\Cref{tab:no_bonus}. Action $a$ is strictly dominant for the opponent. Hence, regardless of the optimizer's action history, the opponent plays~$a$. Therefore, the optimizer cannot benefit from deviating from playing the best response to $a$ ($b$), which is also the Stackelberg leader strategy.
\end{proof}

\begin{table}[htbp]
\centering
\caption{Game where best-responding in every round does not yield the optimal utility for the optimizer (row player).}
\begin{tabular}{c|c|c|}
    & $a$ & $b$ \\
    \hline
    $a$ & $(0, 0)$ & $(0, 1)$ \\
    \hline
    $b$ & $(1, 1 + \frac{\epsilon}{T})$ & $(B, 0)$ \\
    \hline
\end{tabular}
\label{tab:myopic_insufficient}
\end{table}

\begin{proposition}
\label{lem:best-response-is-not-optimal}
    There exist games where myopically best-responding to the opponent's next action does not yield optimal (or near-optimal) utility.
\end{proposition}

\begin{proof}
    Consider the game in~\Cref{tab:myopic_insufficient}. Let $B$ be arbitrarily large. If the optimizer best-responds to the adversary's next move, the optimizer never plays action $a$: $a$ has 0 payoff no matter what the opponent plays, while $b$ always has a positive payoff. Playing action $b$ every round gives the optimizer a payoff of 1 per round. However, if the optimizer plays $ababab...$, they get 0 in each round they play $a$ and $B$ in each round they play $b$. Therefore, the optimizer needs to plan ahead to get the optimal (or even near-optimal) utility.
\end{proof}

In \Cref{lem:best-response-is-not-optimal}, although myopically best-responding to the opponent's next action was not sufficient to achieve the optimal payoff, there was a small cycle ($ab$) which was. It is natural to wonder whether one could always do at least nearly optimally by repeatedly playing some small cycle, a cycle with length polynomial in the number of actions and payoff values. \Cref{lem:no-small-cycle} shows that this is not always possible.

\begin{table}[t]
\centering
\caption{Game where no small cyclical strategy gives optimal utility.}
\begin{tabular}{c|
>{\columncolor{gray!20}}c|
>{\columncolor{gray!20}}c|
c|
>{\columncolor{gray!20}}c|
>{\columncolor{gray!20}}c|
c|}
     & $P_{12}$ & $P_{21}$ & $\cdots$ & $P_{1k}$ & $P_{k1}$ & $\Done$\\
    \hline
    $P_1$ & $(0, p_1)$ & $(0, -p_1)$ & $\cdots$ & $(0, p_1)$ & $(0, -p_1)$ & $(0, \frac{\epsilon}{T})$\\
    \hline
    $P_2$ & $(0, -p_2)$ & $(0, p_2)$ & $(0, 0)$ & $(0, 0)$ & $(0, 0)$ & $(0, \frac{\epsilon}{T})$ \\
    \hline
    $\vdots$ & $(0, 0)$ & $(0, 0)$ & $\ddots$ & $(0, 0)$ & $(0, 0)$ & $(0, \frac{\epsilon}{T})$ \\
    \hline
    $P_k$ & $(0, 0)$ & $(0, 0)$ & $(0, 0)$ & $(0, -p_k)$ & $(0, p_k)$ & $(0, \frac{\epsilon}{T})$ \\
    \hline 
    $\Done$ & $(0, 0)$ & $(0, 0)$ & $(0, 0)$ & $(0, 0)$ & $(0, 0)$ & $(B, -B)$ \\
    \hline
\end{tabular}
\label{tab:no_small_cycle}
\end{table}

\begin{proposition}
\label{lem:no-small-cycle}
    There exist games where no small (i.e., polynomial in the number of actions and the payoff values) repeating cycle of actions can achieve any positive-factor multiplicative approximation or fixed-error additive approximation of the optimal payoff.
\end{proposition}
\begin{proof} Consider the game shown in \Cref{tab:no_small_cycle} for some $k \geq 2$. For each $i \in \{2, \ldots, k\}$, row $P_i$ has corresponding columns $P_{1i}$ and $P_{i1}$; row $P_1$ has entries $(0, p_1)$ and $(0, -p_1)$ in these columns, respectively, and has entry $(0, \frac{\epsilon}{T})$ in column $\Done$. Each other row $P_i$ has entry $(0, -p_i)$ in column $P_{1i}$, $(0, p_i)$ in column $P_{i1}$, $(0, \frac{\epsilon}{T})$ in column $\Done$, and $(0,0)$ elsewhere. All other entries are as shown. Let $p_1, \ldots, p_k$ be distinct primes, let $\epsilon \in (0, 1)$, and let $B>1$ be arbitrarily large. 

The optimizer's payoff is 0 at every entry except at $(\Done, \Done)$, where it is $B$, so any strategy obtaining a positive-factor multiplicative or additive-$C$ approximation must, for $B > C$, induce the opponent to play $\Done$.
    
Let $t$ denote the current round. The opponent plays $\Done$ if and only if $t>1$, all columns besides $\Done$ have score $0$, and $\Done$ has not been played before. When $t=1$, all columns have score $0$ so the opponent plays the leftmost action $P_{12}$. When $t > 1$, if $\Done$ has not been played before, the $\Done$ column has score $(t-1)\cdot \frac{\epsilon}{T} > 0$, so if all other columns have score 0, the opponent will play $\Done$. Suppose in this case another column $P_{ij}$ has a nonzero score. If $P_{ij}$ has a positive score, it must be $\geq 1$ because it is a difference of integers. There are $T$ rounds total, so the $\Done$ column will always have a score $\leq \epsilon < 1$, and the opponent would prefer $P_{ij}$ over $\Done$. If $P_{ij}$ has a negative score, its paired column $P_{ji}$ must have a positive score, so the reasoning above holds for $P_{ji}$. When $t>1$ and all non-$\Done$ columns have score 0, but $\Done$ has been played before, the opponent will not play $\Done$ as it has score $< -B + \epsilon < 0$. 
    
In order to zero out every column $P_{ij}$ while playing at least one non-$\Done$ action, $P_1$ must be played some number of times such that its contribution to each column's score is a common multiple of $p_2, p_3, \ldots, p_k$. Because they are all primes, $P_1$ must therefore be played at least $\prod_{j=2}^{k} p_j$ times. Then, to cancel out the score in each column from playing $P_1$, each other action $P_i$ must be played at least $\frac{\prod_{j=1}^{k} p_j}{p_i}$ times. If we play each action exactly the number of times specified above, we can see that each column will have a score of $\prod_{j=1}^{k} p_j - \prod_{j=1}^{k} p_j = 0$. After this sequence, the opponent will play $\Done$ as its score is positive. The minimum number of rounds required to induce $\Done$ is therefore $1 + \sum_{i=1}^{k} \frac{\prod_{j=1}^{k} p_j}{p_i}$.\footnote{For example, for $k=4$, $(p_1, p_2, p_3, p_4) = (11, 13, 17, 19)$, and any $B > 1$, $12{,}901$ rounds are required, despite the small action space and payoff values.}  
\end{proof}

Interestingly, in contrast to \Cref{lem:no-small-cycle}, in \Cref{section:half_approx} we are able to develop a constant-factor approximation against the memory-limited learning rule \FTLm using a small repeating cycle.
\section{Exact Dynamic Programs for Count-Based Learners}\label{section:dp-algorithms}

In this section, we give exact dynamic programs for finding a sequence that maximizes player~$A$'s expected cumulative payoff against any count-based learning rule $f$ or its memory-limited variant $f_m$. For count-based learners, we focus on games where the number of actions is small; for memory-limited variants, we focus on small $m$. Later, for large $m$, we give an approximation algorithm in \Cref{section:half_approx}. All expectations are over the learner's internal randomness, conditioned on player~$A$'s chosen sequence. Because the learner's mixed action at round $t$ depends only on the history, it suffices to maximize $\sum_{t=1}^{T} \bar{u}^A(a_t,\, f([a_1, \dots, a_{t - 1}]))$ over player~$A$'s action sequence $a_1, \dots, a_T$. Proofs are deferred to~\Cref{section:appendix_dp}.

\subsection{DP for Count-Based Learners} 
\label{section:DP_FTL}

We propose an exact dynamic programming algorithm with runtime $\bigo{{k_A}^2 k_B T^{k_A}}$. Let $V[c_1, \dots, c_{k_A}]$ store the maximum expected cumulative utility for player $A$ over $t = \sum_{i} c_i$ rounds, given that player $A$ played each action $a_i$ exactly $c_i$ times. The table has $\bigo{T^{k_A}}$ entries.

We initialize the table for each entry where the sum of the counts is 1 (corresponding to $t = 1$):
\begin{equation}
\label{dp_ftl:init}
    V[0,\dots, c_i = 1, \dots, 0] = \bar{u}^A(a_i, f(0, \dots, 0)) 
\end{equation}
This is the expected utility of playing action $a_i$ against whichever mixed action the count-based learning rule $f$ chooses when there are no actions that have been played before. 

Next, we fill all entries whose counts sum to $t$, starting from $t = 2$ and going in ascending order to $t = T$, so that $t = \sum_{i = 1}^{k_A} c_i$ and
\begin{align}
    V[c_1,\dots, c_{k_A}] 
    = \max_{i\in [k_A] : c_i \geq 1} \Bigg[&V[c_1,\dots,c_{i}-1,\dots,c_{k_A}]\nonumber\\
    &+ \bar{u}^A(a_i, f(c_1, \dots, c_{i}-1, \dots, c_{k_A}))\Bigg] \label{dp_ftl:recurrence}
\end{align}
The optimal overall expected utility for player $A$ is the maximum entry in $V$ where the sum of all counts adds up to $T$. 

\begin{theorem}
\label{theorem:dp-ftl}
    The Dynamic Program ($\Cref{dp_ftl:recurrence}$) computes the maximum expected cumulative payoff for player $A$ when playing against a count-based learner in time $\bigo{{k_A}^2 k_B T^{k_A}}$. We recover the optimal sequence via an auxiliary array that stores a sequence, which achieves the payoff in $V[c_1, \dots, c_{k_A}]$.
\end{theorem}

\subsection{DP for Memory-Limited Count-Based Learners} 
\label{section:DP_FTL_m}

We give an analogous dynamic programming algorithm to find the maximum expected cumulative payoff for player $A$ when playing against a memory-limited count-based learning rule $f_m$ with runtime $\bigo{{k_A}^{m + 2} k_B T}$. We use a table $V$ parametrized by the time step $t$ and the last $\min(m, t)$ actions played by player~$A$. 

We initialize $V$ for each action $a_i \in \calA$ at time $t = 1$ as follows:
\begin{equation}
    V[1, [a_i]] = \bar{u}^A(a_i, f_m([\emptyset])) \label{dp_ftlm:init}
\end{equation}
This is the expected utility of playing action $a_i$ against the mixed action chosen by the learning rule $f_m$ when no actions have been played before. 

Next, for $t \leq m$ we use the update rule:
\begin{equation}
\label{dp_ftlm:recurrence-less-than-m}
    V[t, [a_1, \dots, a_t]] = V[t - 1, [a_1, \dots, a_{t - 1}]] + \bar{u}^A(a_t, f_m([a_1, \dots, a_{t - 1}])) 
\end{equation}
While, for $t > m$ the state $[a_1, ..., a_m]$ denotes the last $m$ actions of player $A$, ordered from least to most recent, and we use the update rule:
\begin{align}
    V[t, [a_1, \dots, a_m]] = \max_{a \in \calA} \Bigg[&V[t - 1, [a, a_1, \dots, a_{m - 1}]] \nonumber\\
    &+ \bar{u}^A(a_m, f_m([a, a_1, \dots, a_{m - 1}]))\Bigg] \label{dp_ftlm:recurrence-more-than-m}
\end{align}
The optimal expected cumulative utility for player $A$ is in the maximum entry of~$V$ where $t = T$.

\begin{theorem}\label{theorem:dp-ftl-m}
    The Dynamic Program in $\Cref{dp_ftlm:recurrence-more-than-m}$ computes the maximum expected cumulative payoff for player $A$ against a memory-limited count-based learner in time $\bigo{{k_A}^{m + 2} k_B T}$. We recover the optimal sequence via an auxiliary array that stores a sequence, which achieves the payoff in $V[t, [a_1, \dots, a_m]]$. 
\end{theorem}

\section{Approximation Algorithms for Memory-Limited Count-Based Learners}
\label{section:approximation-algorithms}

The exact algorithm presented previously for memory-limited count-based learners is practical when $m$ is small and $T$ is not too large. In this section, we develop three approximation algorithms 
for memory-limited count-based learners when $m$ is large, or when $m$ is small but $T$ is very large. All three algorithms assume that player $A$'s payoffs lie in $[0, H]$. Their guarantees and use cases are summarized in \Cref{table:summary}, alongside the previous exact result for comparison. 

\begin{table}[tbp]
\centering
\caption{Summary of algorithmic results for memory-limited count-based learners. The first row summarizes the previous exact result, while the remaining rows summarize the approximation results from this section. $\ALG$ denotes the cumulative expected payoff achieved by the corresponding algorithm for player~$A$, and $\OPT$ is the best cumulative expected payoff that can be achieved by player~$A$. Both expectations are over the opponent’s internal randomness. $\varepsilon$ is a tunable accuracy parameter.}
\label{table:summary}
\begin{tblr}{
    width = \linewidth,
    colspec = {lllX[l]},
    row{1} = {font=\bfseries},
    row{2} = {bg=gray!10},
    cells = {valign=m, halign=l},
    hline{1} = {1pt, solid},
    hline{2} = {0.5pt, solid},
    hline{Z} = {1pt, solid},
    rowsep = {2.5pt},
    colsep = {4pt}
}
    Result & Guarantee & Running Time & Use Case \\
    Thm. \ref{theorem:dp-ftl-m} & $\ALG = \OPT$ & $\bigo{{k_A}^{m + 2} k_B T}$ & $m$ small, $T$ not too large \\
    Thm. \ref{thm:half-approx} & $\ALG \geq \OPT / 2$ & $\bigo{k_B k_A^3 (m + 2k_A)^{2k_A}}$ & $m$ large \\
    Thm. \ref{thm:ptas} & $\ALG \geq (1 - \varepsilon) (\OPT - mH)$ & $\bigo{k_A^{\lceil m / \varepsilon \rceil + 1}k_B \lceil m / \varepsilon \rceil }$ & $m$ very small, $T$ very large, and want better than 1/2-approximation  \\
    Thm. \ref{thm:additive-approx} & $\ALG \geq \OPT - \bigo{k_A^m H}$ & $\bigo{k_A^{m} (m + k_A^{m + 1} + k_A k_B)}$ & $m$ small, $T$ very large, and the additive loss is acceptable\\
\end{tblr}
\end{table}

\subsection{\texorpdfstring{$1/2$}{1/2}-Approximation Algorithm}
\label{section:half_approx}

We propose an algorithm (\Cref{alg:half-approx}) with runtime $\bigo{k_B k_A^3(m+2k_A)^{2k_A}}$ that guarantees at least half of player $A$'s optimal expected cumulative payoff. We assume $m < T$; otherwise, one can use~\Cref{dp_ftl:recurrence} directly. 
We defer proofs to~\Cref{appendix:proofs-from-half-approx}.

The algorithm partitions the horizon into blocks of length $m$ (the last block has length $p = T \bmod m$). The key observation is that $A$ gets at least half of $\OPT$ in either the odd- or even-numbered blocks, so optimizing each and taking the better yields a $1/2$-approximation. The algorithm requires three subsequences of actions: $S_m$, a sequence maximizing $A$'s expected payoff over the first $m$ time steps (computed with~\Cref{dp_ftl:recurrence}); $S_{2m}$, a sequence of length $2m$ maximizing $A$'s expected payoff in the last $m$ time steps; and $S_{m+p}$, a sequence of length $m + p$ maximizing $A$'s expected payoff in the last $p$ time steps. Below we give a Dynamic Program for computing $S_{2m}$ and $S_{m + p}$. We leave it as an open question whether a better approximation factor can be achieved in time polynomial in~$m$.

\paragraph{Dynamic Program for computing $S_{2m}$ and $S_{m + p}$.}
Let $V[c_1,\dots,c_{k_A}][c'_1,\dots,c'_{k_A}]$ be the expected payoff of $A$ in the interval $(m,m+\sum_i c'_i]$ where in the time-interval $(m-\sum_i c_i,m]$ each action $a_i\in \calA$ is played $c_i$ times, and in the interval $(m, m+\sum_i c'_i]$, each action $a_i\in \calA$ is played $c'_i$ times. We initialize the table for each entry where the sum of the counts $c_1,\dots, c_{k_A}$ is equal to $m$:
\begin{equation}
    V[c_1,\dots,c_{k_A}][0,\dots,0]=0 \text{ where } \sum_{i=1}^{k_A} c_{i} = m
\end{equation}

We fill in the remaining entries according to the following recurrence. Start from $t = m + 1$ and fill out the entries in ascending order up to $2m$, where $t = m + \sum_{i=1}^{k_A} c'_i$. To compute the maximum expected payoff received by $A$ in the interval $(m, t]$, we iterate over all pairs of actions $(a_i, a_j)$ where $a_i$ was played at time step $t - m$ and $a_j$ is played at step $t$ by player $A$. 
For a fixed pair $(a_i, a_j)$, $A$'s maximum expected payoff at step $t$ is computed by adding her expected payoff in the current step $t$ where the opponent is a memory-limited count-based learning rule $f_m$ with a count history of $(c_1+c'_1,\dots, c_i+c'_i+1,\dots, c_j+c'_j-1+\dots+c_{k_A}+c'_{k_A})$, plus $A$'s expected payoff in step $t-1$ where the history of her actions is $(c_1,\dots,c_i+1,\dots,c_{k_A}, c'_1,\dots, c'_j-1,\dots,c'_{k_A})$:
\begin{align}
    V[c_1,&\dots,c_{k_A}][c'_1,\dots,c'_{k_A}] \nonumber\\
    &= \max_{\substack{(i, j) \in [k_A]^2 \\ c'_j \geq 1}}\Big[V[c_1,\dots,c_i+1,\dots,c_{k_A}][c'_1,\dots,c'_j-1,\dots,c'_{k_A}] \nonumber\\
    &\ + \bar{u}^A(a_j, f_m(c_1+c'_1,\dots,c_i+c'_i+1, \dots,  c_j+c'_j-1, \dots c_{k_A}+c'_{k_A}))\Big]
    \label{dp:1/2-approx}
\end{align}

\begin{algorithm}[t]
    \caption{\texorpdfstring{$1/2$}{1/2}-approximation algorithm}
    \label{alg:half-approx}
    
    \raggedright
    \setcounter{AlgoLine}{0}
    
    \KwIn{Opponent's count-based learning rule $f_m$, memory $m$, action spaces $\mathcal{A}$ and $\mathcal{B}$, time horizon $T$}
    \KwOut{Action sequence for player $A$}
    
    Divide the time horizon into blocks of size $m$, except for the last block that has size $p = T \bmod m$\;
    
    \tcc{Compute the three optimal sequences below using DP}
    $S_{m}$ $\leftarrow$ sequence maximizing $A$'s expected payoff in first $m$ time steps\;
    
    $S_{2m}$ $\leftarrow$ sequence of length $2m$ maximizing $A$'s expected payoff in last $m$ time steps\;
    
    $S_{m+p}$ $\leftarrow$ sequence of length $m+p$ maximizing $A$'s expected payoff in last $p$ time steps\;
    
    \tcc{Player $A$'s payoff is maximized in either odd or even blocks}
    
    Consider these two strategies and play the one with higher expected payoff:
    
    \textbf{Odd strategy:} Play $S_m$, then repeatedly play $S_{2m}$. If a block of length $m+p$ remains, play $S_{m+p}$. Otherwise, play arbitrary actions until time $T$\;
    
    \textbf{Even strategy:} Repeatedly play $S_{2m}$. If a block of length $m+p$ remains, play $S_{m+p}$. Otherwise, play arbitrary actions until time $T$\;
\end{algorithm}

To compute $S_{2m}$, we note that the first count vector $c_1, \dots, c_{k_A}$ must be fully exhausted, so the optimal payoff over the last $m$ time steps is $\max_{\sum_i c'_i = m} V[0, \dots, 0][c'_1, \dots, c'_{k_A}]$. 
Because $m - p$ actions remain in the first half for $S_{m + p}$, the optimal payoff over the last $p$ time steps is $\max_{\sum_i c_i = m - p, \sum_i c'_i = p} V[c_1, \dots, c_{k_A}][c'_1, \dots, c'_{k_A}]$.
In both cases, the corresponding sequence is recovered by backtracking through the table and selecting the optimal action at each step.

\begin{algorithm}[t]
    \caption{$(1-\eps)(\OPT-mH)$-approximation algorithm}
    \label{alg:PTAS}
    \raggedright
    
    \KwIn{Opponent's count-based learning rule $f_m$, memory $m$, action spaces $\mathcal{A}$ and $\mathcal{B}$, time horizon $T$, approximation factor $\eps$}
    \KwOut{Action sequence for player $A$}
    
    \SetAlgoRefRelativeSize{-1}
    \setcounter{AlgoLine}{0}
    
    \tcc{Phase 1: Find optimal sequence for first block}
    $\mathcal{S} \leftarrow$ set of all possible action sequences in the first $\lceil m/\eps \rceil$ time steps\;
    
    $s' \leftarrow \argmax_{s\in \mathcal{S}} \sum_{t=\min(m+1, T)}^{\min(\lceil m/\eps \rceil, T)} \bar{u}^A(s_t, b_t)$ where $b_t = f_m([s_{\max(1,t-m)},\dots,s_{t-1}])$ when Player $A$ plays sequence $s$\;
    
    \tcc{Phase 2: Repeat optimal sequence for full blocks}
    \For{$i = 1, \ldots, \left\lfloor \frac{T}{\lceil m/\eps \rceil} \right\rfloor$}{
        $A$ plays the action sequence $s'$ for time steps $\left((i-1)\lceil \frac{m}{\eps}\rceil, i\lceil\frac{m}{\eps}\rceil\right]$\;
    }
    
    \tcc{Phase 3: Handle remaining time steps}
    $\mathcal{S}' \leftarrow$ all possible action sequences for remaining time steps $\left(\left\lfloor \frac{T}{\lceil m/\eps\rceil} \right\rfloor \lceil \frac{m}{\eps}\rceil, T\right]$\;
    
    $s' \leftarrow \argmax_{s\in \mathcal{S}'} \sum_{t = \min\left(\left\lfloor \frac{T}{m/\eps} \right\rfloor \lceil \frac{m}{\eps} \rceil+m, T\right)}^{T} \bar{u}^A(s_t, b_t)$\;
    
    Play sequence $s'$ for remaining time steps\;
\end{algorithm}

\begin{theorem}
\label{thm:DP-S-2m}
    The Dynamic Program in \Cref{dp:1/2-approx} finds a sequence of length $2m$ that maximizes $A$'s expected payoff in the last $m$ time steps against a memory-limited count-based learning rule $f_m$. It runs in $\bigo{k_Bk_A^3(m+2k_A)^{2k_A}}$ time.
\end{theorem}

\begin{theorem}
\label{thm:half-approx}
    \Cref{alg:half-approx} guarantees a $1/2$-approximation for $A$'s expected cumulative payoff in time $\bigo{k_Bk_A^3(m+2k_A)^{2k_A}}$.
\end{theorem}

\subsection{\texorpdfstring{$(1-\eps)(\OPT-mH)$}{(1-eps)(OPT-mH)}-Approximation Algorithm} 
\label{section:ptas}

Our next approximation algorithm (\Cref{alg:PTAS}) computes a sequence for player $A$ guaranteeing expected payoff at least $(1-\eps)(\OPT - mH)$ against a memory-limited count-based learning rule $f_m$, in time $O(k_A^{m/\eps + 1} k_B (m/\eps))$. The algorithm finds the sequence of length $(m/\eps)$ maximizing payoff over its last $m/\eps - m$ time steps, then repeats it for the remaining time steps. This algorithm is particularly useful when $m$ is very small but $T$ is very large. We defer proofs to~\Cref{appendix:proofs-from-ptas}.

\begin{theorem}
\label{thm:ptas}
    \Cref{alg:PTAS} guarantees expected cumulative payoff of $(1-\eps) (\OPT-mH)$ for player $A$ in time $O(k_A^{\lceil m/\eps \rceil + 1} k_B \lceil m/\eps \rceil)$.
\end{theorem}

\subsection{\texorpdfstring{$\bigo{k_A^m H}$}{O(k\_A\^m H)}-Additive Approximation Algorithm}
\label{section:additive_approx}

Finally, we present a complementary algorithm for any finite memory $m \geq 1$ with expected cumulative payoff of at least $\OPT - \bigo{k_A^m H}$. As $T \to \infty$, this additive penalty becomes asymptotically negligible. The algorithm runs in time $\bigo{k_A^{2m + 1} + m k_A^{m + 1} k_B}$. Since the runtime is independent of $T$ but is exponential in $m$, this algorithm is again useful when $m$ is small and $T$ is large. When $m$ and $T$ are both small, the exact DP (\Cref{dp_ftlm:recurrence-more-than-m}) is preferable; when $m$ and $T$ are both large, the $1/2$-approximation (\Cref{alg:half-approx}) is the practical choice.

We modify the approach of repeating action sequences that we already used in \Cref{section:half_approx,section:ptas}. We model the memory-limited learning rule as a state space over sequences of $m$ actions. Karp's maximal mean weight cycle algorithm~\cite{karp1978characterization,chaturvedi2017note} allows us to extract an optimal repeating sequence of actions. Repeating this optimal cycle gives that the optimizer only incurs an additive utility loss at the prefix and suffix of the time horizon. Because the state space has size $k_A^m$, the maximum length of any simple cycle is bounded by $k_A^m$, so the algorithm loses at most $\bigo{k_A^m H}$ in utility. The formal construction and proofs are deferred to~\Cref{section:additive-approx-proofs}.

\begin{theorem}
\label{thm:additive-approx}
    There exists an algorithm that guarantees an expected cumulative payoff of at least $\OPT - \bigo{k_A^m H}$ for Player $A$ against a memory-limited count-based learning rule $f_m$, which runs in $\bigo{k_A^{2m + 1} + k_A^{m} (m + k_A k_B)}$ time.
\end{theorem}

For instance, for $m = 1$ \Cref{thm:additive-approx} gives a $\bigo{k_A H}$-additive approximation in $\bigo{k_A^3 k_B}$ time.

\section{Deceiving an Unknown Opponent}
\label{sec:unknown-opponent}
    Suppose the optimizer does not know the particular learning rule that the opponent is using, but does know that it belongs to some family $\mathcal{F}$ of count-based learning rules, each with memory length $m$. Can the optimizer learn to deceive their opponent in this setting? 

    A natural idea is to try to quickly determine which opponent we are facing, and then play optimally against that opponent. Since a count-based opponent with memory length $m$ chooses their next action based on only the counts of our previous $m$ actions, we might intuitively try something like the following:
    
    \begin{enumerate}
        \item while we have not yet determined the true opponent: 
        \begin{enumerate}
            \item play any $m$ actions against which at least two opponents have a different response
            \item observe the response, eliminating at least one possible opponent
        \end{enumerate}
        \item compute and play the optimal sequence against the true opponent for the remaining rounds
    \end{enumerate}

    With some assumptions (similar to those below) this idea can be used for deterministic count-based opponents. The learning process would require up to $O(m \abs{\calF})$ rounds, after which we can use a DP to compute the optimal sequence against the true opponent and play that sequence for the remaining rounds.

    However, we would like to handle \textit{any} count-based opponent, including randomized opponents. In that case, we cannot necessarily eliminate any opponent by simply observing the response to a full history: even if two randomized opponents have a different response distribution to a sequence of $m$ actions, it may take many repetitions of that sequence to have high confidence about which we are playing against. This motivates our approach, which leverages existing results for stochastic bandits; this is a natural connection as in both settings we must balance learning about unknown reward distributions against maximizing overall payoff. 

    We develop a 1/2-approximation up to $O(\sqrt{|\mathcal{F}|mT} + m$) regret compared to the optimal cumulative payoff against the actual opponent (\Cref{alg:unknown-opponent}), under two assumptions. We assume there exists a ``blank'' action for which both players always have zero payoff; when played $m$ times, it effectively clears the history, restoring the opponent's behavior to what it would be at the beginning of the game. We also assume that all payoffs for the optimizer lie in $[0, 1]$. This algorithm is practical when $m$ and $|\mathcal{F}|$ are small relative to $T$ (it is an asymptotic 1/2-approximation when we consider them constants).

\subsection{1/2-Approximation up to $O(\sqrt{|\mathcal{F}|mT} + m$) Regret} 

At a high level, this algorithm combines a similar idea to our 1/2-approximation (\Cref{alg:half-approx}) with a reduction to stochastic bandits. 

\begin{algorithm}[t]
\caption{$1/2$-approximation up to $O(\sqrt{|\mathcal{F}|mT} + m$) regret against an unknown opponent from count-based family $\mathcal{F}$ with memory $m$}
\label{alg:unknown-opponent}
\raggedright
\setcounter{AlgoLine}{0}

\KwIn{Family of possible opponent learning rules $\mathcal{F}$ (each count-based with memory length $m$), stochastic bandits algorithm $\mathsf{Alg_{SB}}$}
\For{$f$ in $\mathcal{F}$}{
    $S_{f2m} \gets$ sequence maximizing $A$'s expected payoff in the last $m$ time steps if the opponent is using learning rule $f$
}

$T' \gets \lfloor \frac{T}{2m}\rfloor$

Run $\mathsf{Alg_{SB}}$ over $T'$ rounds, where we have $|\mathcal{F}|$ arms, pulling arm $i$ is equivalent to playing $S_{f2m}$ and the realized payoff is the realized payoff over the second $m$ steps of $S_{f2m}$. 

Play arbitrary actions for the remaining $T-2mT'$ time steps.
\end{algorithm}

The guarantee stated in \Cref{thm:unknown-opponent} depends on the specific stochastic bandits algorithm used in \Cref{alg:unknown-opponent}; we assume the use of MOSS from \cite{stochastic_bandits} for concreteness and the simplicity of the bound. \cite{stochastic_bandits} derive a regret bound of $O(\sqrt{KT})$ for the $K$-armed stochastic bandit problem with bounded rewards. The proof of \Cref{thm:unknown-opponent} is deferred to \Cref{appendix:missing-proof-unknown-opponent}. 

\begin{theorem}
\label{thm:unknown-opponent}
    \Cref{alg:unknown-opponent} gives a 1/2 approximation up to $O(\sqrt{|\mathcal{F}|mT} + m$) regret to the optimal expected payoff against the actual opponent. 
\end{theorem}

\begin{remark}
    When $m$ and $|\mathcal{F}|$ do not scale with $T$, \Cref{alg:unknown-opponent} achieves an asymptotic 1/2  approximation. In that case, the regret is sublinear in $T$. Given the opponent's bounded memory and the existence of the ``blank'' action, any instance with a nonzero optimal value satisfies $\OPT \in \Theta(T)$, so the regret term vanishes in comparison as $T \rightarrow \infty$. If $\OPT=0$ the guarantee holds trivially since all payoffs are non-negative.
\end{remark}

\section{Hardness Results}
\label{section:hardness-results}

In this section we show that computing the optimal action sequence against \FTL is NP-hard to approximate. The proof is deferred to \Cref{appendix:missing-proofs-hardness}.

\begin{theorem}
\label{thm:np_hardnes}
    It is NP-hard to compute a sequence of actions which attains payoff within any constant multiplicative factor or any $T^\alpha$ additive error (for constant $0 < \alpha < 1$) of the optimal payoff against \FTL, even in zero-sum games.
\end{theorem}

\Cref{thm:np_hardnes} improves upon a hardness result of Assos et al.~\cite{assosmaximizing}, who show inapproximability within constant additive error in general-sum games. In contrast, our result holds \emph{even} in the case of zero-sum games and strengthens the additive bound to $T^\alpha$ for any constant $0 < \alpha < 1$.
In a related line of work, Assos et al.~\cite{assos2025computational} give an $\bigomega{T}$ additive hardness bound for a different opponent learning algorithm---Multiplicative Weights Update.
\section{Experiments}

As we showed in \Cref{sec:structural_observations}, depending on the game, the deception bonus can be arbitrarily large (\Cref{lem:bonus-is-large}), or there might be no gains from deception at all (\Cref{lem:no-deception-bonus}). To explore what the gains from deception might typically look like, we empirically measured the deception bonus in random zero-sum and general-sum games. The deception bonus is computed with respect to a particular opponent (see \Cref{definition:deception_bonus}); for concreteness in these experiments, we used $\FTL$. 

\paragraph{Experimental design}
For a random zero-sum game (described by payoff matrices $A$ and $B$ for each player), payoffs for $A$ are drawn independently and uniformly at random from the real interval $[-1, 1]$, and $B=-A$. For a random general-sum game, payoffs for both $A$ and $B$ are drawn independently and uniformly at random from $[-1, 1]$. Each experiment (corresponding to a point on the plots below) was run for 500 independent trials. \footnote{Source code: \url{https://github.com/dimitarch/deception-code}}

\begin{figure}[tb]
    \centering
    \begin{minipage}[t]{0.48\linewidth}
        \centering
        \includegraphics[width=\linewidth, trim=0mm 0mm 0mm 0mm, clip]{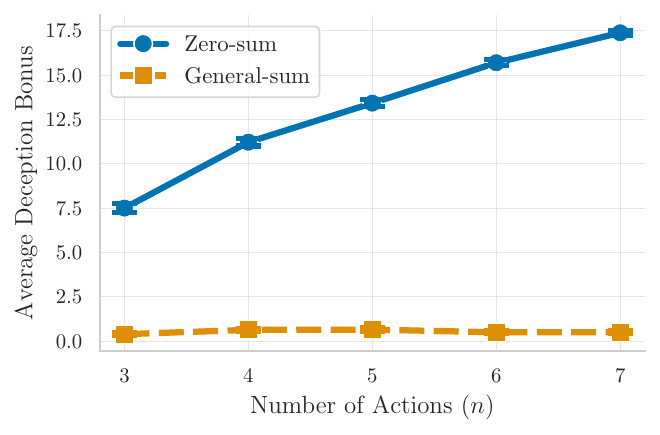}
        \caption{Deception bonus vs.\ number of actions for $T = 25$ rounds. Error bars show standard error over 500 trials.}
        \label{figure:bonus-vs-actions-t25}
    \end{minipage}
    \hfill
    \begin{minipage}[t]{0.48\linewidth}
        \centering
        \includegraphics[width=\linewidth, trim=0mm 0mm 0mm 0mm, clip]{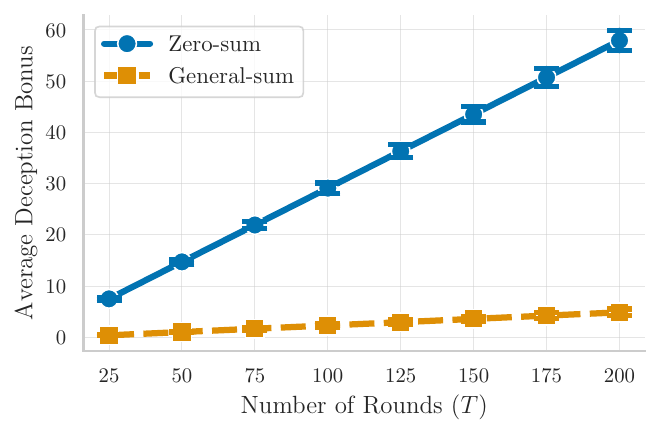}
        \caption{Deception bonus vs.\ number of rounds for $n = 3$ actions. Error bars show standard error over 500 trials.}
        \label{figure:bonus-vs-time-n3}
    \end{minipage}
\end{figure}

\paragraph{Findings} We summarize our primary findings below.
\begin{itemize}
    \item The average deception bonus is substantially larger in random zero-sum games than in random general-sum games, where it was consistently low (\Cref{figure:bonus-vs-actions-t25,figure:bonus-vs-time-n3}). This suggests that on average deception tends to be lucrative in competitive settings while providing little advantage when opportunities for collaboration exist.

    \item  In zero-sum games the deception bonus increases with the number of actions $n$  (\Cref{figure:bonus-vs-actions-t25}). One explanation is that in a random game, as $n$ grows it is more likely that there exists a high-payoff response to each of the opponent's actions; the optimizer can take advantage of this due to the opponent's determinism. In general-sum games, however, the deception bonus remains relatively flat as $n$ increases. This may be because larger action spaces also increase the probability of mutually beneficial actions, which reduce the advantage of deceptive play.

    \item In both zero-sum and general-sum games, the deception bonus is roughly linear in the number of rounds (\Cref{figure:bonus-vs-time-n3}), indicating that the per-round gains from deception remain relatively constant for a fixed $n$ as the number of rounds varies.
\end{itemize}

\section{Discussion and Future Work}

We introduced a new game-theoretic model of deception in repeated games. We provided structural results, a hardness of approximation result for optimizing against Empirical Risk Minimization which holds even in zero-sum games, and experimental results on the deception bonus in random games. We also presented a variety of algorithmic results for optimizing against a count-based opponent: exact dynamic programming algorithms for when the game or the opponent’s memory is small, three approximation algorithms---a $1/2$-approximation with a time-complexity polynomial in the opponent's memory $m$, an asymptotic PTAS and an additive-approximation whose runtimes are exponential in $m$---and an algorithm for learning to optimize against an unknown count-based opponent. Our model captures deception over time, which we believe is an interesting and underexplored direction; we hope it will lead to further work in this area.

\paragraph{Open questions and future work}

It remains an open question whether our $1/2$-approximation result can be improved while keeping the runtime polynomial in $m$. Furthermore, while we present an initial result on learning to deceive an opponent whose learning rule is unknown, it would be interesting to explore whether stronger results can be obtained in that setting---perhaps by leveraging knowledge about the opponent’s possible learning rules that is not exploited by our reduction to stochastic bandits. Finally, it remains open whether (approximately) optimizing against $\FTL$ is still NP-hard in the special case of zero-sum games in which all payoffs are integers bounded by $\littleo{T}$. 

\begin{credits}
\subsubsection{Acknowledgements} This work was supported in part by the National Science Foundation under grants CCF-2212968 and ECCS-2216899 and through an NSF CSGrad4US Fellowship under grant 2313998, by the Simons Foundation under the Simons Collaboration on the Theory of Algorithmic Fairness, and by the Office of Naval Research MURI Grant N000142412742. The
views expressed in this work do not necessarily reflect the position or the policy of the Government and no official endorsement should be inferred.

\end{credits}

\bibliographystyle{splncs04}
\bibliography{ref}

@inproceedings{gan2019imitative,
  title={Imitative follower deception in stackelberg games},
  author={Gan, Jiarui and Xu, Haifeng and Guo, Qingyu and Tran-Thanh, Long and Rabinovich, Zinovi and Wooldridge, Michael},
  booktitle={EC},
  pages={639--657},
  year={2019}
}

@article{birmpas2020optimally,
  title={Optimally deceiving a learning leader in stackelberg games},
  author={Birmpas, Georgios and Gan, Jiarui and Hollender, Alexandros and Marmolejo, Francisco and Rajgopal, Ninad and Voudouris, Alexandros},
  journal={NeurIPS},
  volume={33},
  pages={20624--20635},
  year={2020}
}

@inproceedings{assosmaximizing,
  title={Maximizing utility in multi-agent environments by anticipating the behavior of other learners},
  author={Assos, Angelos and Dagan, Yuval and Daskalakis, Constantinos},
  booktitle={NeurIPS},
  year={2024}
}

@inproceedings{assos2025computational,
  title={Computational Intractability of Strategizing against Online Learners},
  author={Assos, Angelos and Dagan, Yuval and Rajaraman, Nived},
  booktitle={COLT},
  pages={169--199},
  year={2025},
  organization={PMLR}
}

@inproceedings{fortnow1994optimality,
  title={Optimality and domination in repeated games with bounded players},
  author={Fortnow, Lance and Whang, Duke},
  booktitle={STOC},
  pages={741--749},
  year={1994}
}

@inproceedings{freund1995efficient,
  title={Efficient algorithms for learning to play repeated games against computationally bounded adversaries},
  author={Freund, Yoav and Kearns, Michael and Mansour, Yishay and Ron, Dana and Rubinfeld, Ronitt and Schapire, Robert E},
  booktitle={FOCS},
  pages={332--341},
  year={1995},
  organization={IEEE}
}

@inproceedings{nguyen2020tackling,
  title={Tackling Imitative Attacker Deception in Repeated Bayesian Stackelberg Security Games.},
  author={Nguyen, Thanh Hong and Butler, Andrew and Xu, Haifeng},
  booktitle={ECAI},
  pages={187--194},
  year={2020}
}

@inproceedings{nguyen2019imitative,
  title={Imitative Attacker Deception in Stackelberg Security Games.},
  author={Nguyen, Thanh and Xu, Haifeng},
  booktitle={IJCAI},
  pages={528--534},
  year={2019}
}

@article{karp1978characterization,
  title={A Characterization of the Minimum Cycle Mean in a Digraph},
  author={Karp, Richard M},
  journal={Discrete Mathematics},
  volume={23},
  number={3},
  pages={309--311},
  year={1978},
  publisher={Elsevier}
}

@article{dasdan2002faster,
  title={Faster maximum and minimum mean cycle algorithms for system-performance analysis},
  author={Dasdan, Ali and Gupta, Rajesh K},
  journal={IEEE Transactions on Computer-aided Design of Integrated Circuits and Systems},
  volume={17},
  number={10},
  pages={889--899},
  year={2002}
}

@article{chen2023learning,
  title={Learning to Manipulate a Commitment Optimizer},
  author={Chen, Yurong and Deng, Xiaotie and Gan, Jiarui and Li, Yuhao},
  journal={arXiv preprint arXiv:2302.11829},
  year={2023}
}

@inproceedings{balcan2018diversified,
  title={Diversified Strategies for Mitigating Adversarial Attacks in Multiagent Systems.},
  author={Balcan, Maria-Florina and Blum, Avrim and Chen, Shang-Tse},
  booktitle={AAMAS},
  pages={407--415},
  year={2018}
}

@article{deng2019strategizing,
  title={Strategizing against No-regret Learners},
  author={Deng, Yuan and Schneider, Jon and Sivan, Balasubramanian},
  journal={NeurIPS},
  volume={32},
  year={2019}
}

@article{brown2023learning,
  title={Is Learning in Games Good for the Learners?},
  author={Brown, William and Schneider, Jon and Vodrahalli, Kiran},
  journal={NeurIPS},
  volume={36},
  pages={54228--54249},
  year={2023}
}

@inproceedings{braverman2018selling,
  title={Selling to a No-regret Buyer},
  author={Braverman, Mark and Mao, Jieming and Schneider, Jon and Weinberg, Matt},
  booktitle={EC},
  pages={523--538},
  year={2018}
}

@InProceedings{mansour2022strategizing,
  title = 	 {Strategizing against Learners in Bayesian Games},
  author =       {Mansour, Yishay and Mohri, Mehryar and Schneider, Jon and Sivan, Balasubramanian},
  booktitle = 	 {COLT},
  pages = 	 {5221--5252},
  year = 	 {2022},
  editor = 	 {Loh, Po-Ling and Raginsky, Maxim},
  volume = 	 {178},
  series = 	 {PMLR}
}

@inproceedings{aitchison2021learning,
  title={Learning to Deceive in Multi-agent Hidden Role Games},
  author={Aitchison, Matthew and Benke, Lyndon and Sweetser, Penny},
  booktitle={Deceptive AI (DeceptAI)},
  pages={55--75},
  year={2021},
  organization={Springer}
}

@inproceedings{carminati2024hidden,
  title={Hidden-Role Games: Equilibrium Concepts and Computation},
  author={Carminati, Luca and Zhang, Brian Hu and Farina, Gabriele and Gatti, Nicola and Sandholm, Tuomas},
  booktitle={EC},
  pages={106--107},
  year={2024}
}

@inbook{Kuhn+1951+97+104,
title = {A Simplified Two-Person Poker},
booktitle = {Contributions to the Theory of Games, Volume I},
author = {H. W. Kuhn},
editor = {Harold William Kuhn and Albert William Tucker},
publisher = {Princeton University Press},
address = {Princeton},
pages = {97--104},
doi = {10.1515/9781400881727-010},
isbn = {9781400881727},
year = {1951},
lastchecked = {2025-07-01}
}

@inproceedings{stochastic_bandits,
  title={Minimax Policies for Adversarial and Stochastic Bandits},
  author={Audibert, Jean-Yves and Bubeck, S{\'e}bastien},
  booktitle={COLT},
  pages={217--226},
  year={2009}
}

@article{chaturvedi2017note,
  title={A note on finding minimum mean cycle},
  author={Chaturvedi, Mmanu and McConnell, Ross M},
  journal={Information Processing Letters},
  volume={127},
  pages={21--22},
  year={2017},
  publisher={Elsevier}
}
%

% Appendix
\appendix
\crefalias{section}{appendix}
\crefalias{subsection}{subappendix}
\crefalias{subsubsection}{subsubappendix}
\section{Proofs from \Cref{section:dp-algorithms}} \label{section:appendix_dp}

\begin{proof}[Proof of \Cref{theorem:dp-ftl}]
    For $t = 1$, each count vector is a one-hot vector representing a single action $a_i$ being played once. The best possible score corresponding to $a_i$ being played in the first round is the expected payoff of $a_i$ against the action the opponent chooses given no history. This corresponds to the initialization in~\Cref{dp_ftl:init}.
    
    Assume that all entries corresponding to count vectors that sum to $1, \dots, t - 1$, have been correctly computed and hold the maximum expected score that can be achieved with those counts. We will show that each entry corresponding to a count vector summing to $t$ will also be computed correctly. At time $t$, the best overall expected payoff we can get corresponds to the expected payoff we get this round for playing a particular action $a_i$ which has count $c_i \geq 1$, plus the maximum expected payoff we could get in the previous $t - 1$ rounds from playing the other $t - 1$ actions represented by the count vector $(c_1, \dots, c_i - 1, \dots, c_{k_A})$ (the count-based learning rule $f$ is unaffected by the order of play, so we can simply choose the order with the highest expected payoff). By the inductive assumption, the maximum expected payoff for the previous $t - 1$ actions is already correctly computed, so we can simply take the maximum over the possible action choices for time $t$. This corresponds to the update rule in~\Cref{dp_ftl:recurrence}.

    The size of the DP table is $\bigo{T^{k_A}}$. Computing each entry takes $\bigo{{k_A}^2 k_B}$ time: there are $k_A$ actions of player $A$ to go over when taking the max, and for each we compute the expected payoff under $f$ in $O(k_Ak_B)$ time. Note that for computing $f$, we assume we keep track of the action counts as we go to avoid an $\bigo{T}$ factor of recomputing the counts for past rounds. Therefore the total time complexity is $\bigo{{k_A}^2 k_B T^{k_A}}$.
\end{proof}

\begin{proof}[Proof of \Cref{theorem:dp-ftl-m}]
    For $t = 1$, the value in $V[1, [a_i]]$ is exactly the expected utility that Player $A$ gets when playing $a_i$ at the very beginning of the game.

    Assume that the table $V$ has been correctly computed up to time $t - 1$ for every possible sequence of the last $m$ actions. At time $t$ fix the sequence of the last $m$ actions $[a_1, \dots, a_m]$. Then the expected cumulative payoff we can get within a $t$-round game that ends with $[a_1, \dots, a_m]$ is the maximum over all actions $a$ that can precede $[a_1, \dots, a_m]$ of the sum of the expected utility of playing $a_m$ against $f_m([a, a_1, \dots, a_{m - 1}])$ and the best expected score for a $t - 1$ round game that ends in $[a, a_1, \dots, a_{m - 1}]$, i.e. $V[t - 1, [a, a_1, \dots, a_{m - 1}]]$. By the induction hypothesis, those entries are computed correctly, so it must be that $V[t, [a_1, \dots, a_m]]$ is also computed correctly.

    The DP table has size $\bigo{T {k_A}^m}$: there are $T$ time steps and the number of distinct sequence vectors is $\bigo{{k_A}^m}$. Each entry takes $\bigo{{k_A}^2 k_B}$ time to compute: there are $k_A$ actions of player $A$ to go over when taking the max, and for each we computing the memory-limited count-based learning rule $f_m$ in $O(k_Ak_B)$ time. Similarly to the proof of~\Cref{theorem:dp-ftl}, for computing $f_m$, we assume we keep track of the action counts as we go to avoid an $\bigo{m}$ factor of recomputing the counts for past rounds. Therefore, the total time complexity is $\bigo{{k_A}^{m + 2} k_B T}$.
\end{proof}

\section{Proofs from \Cref{section:half_approx}}
\label{appendix:proofs-from-half-approx}

\begin{proof}[Proof of \Cref{thm:DP-S-2m}]

\smallskip
\noindent\textbf{Correctness.\ } 
We prove the correctness of the Dynamic Programming Algorithm by induction. 

\smallskip
\noindent\textbf{Base case.\ } 
For $t=m$, and any count vector $(c_1,\cdots,c_{k_A})$ that add up to $m$, the value in \[V[c_1,\cdots,c_{k_A}][0,\cdots,0]=0,\] since no actions are played in the second $m$ time steps yet. 

Assume that the table $V$ has been correctly computed up to time $t'<t$, for any count vector $(c_1,\cdots,c_{k_A},c'_1,\cdots,c'_{k_A})$ where $t'=m+\sum_{i=1}^{k_A}c'_i$ and $\sum_i c_i + \sum_i c'_i = m$. 
We reiterate the recurrence for the DP in the following. In order to compute the maximum expected payoff received by $A$ in the interval $(m,t]$, first we guess a pair of actions $(a_i, a_j)$ where $a_i$ was played at time step $t-m$ and $a_j$ is played at step $t$ by player $A$. For a fixed pair $(a_i, a_j)$, $A$'s maximum expected payoff at step $t$ is computed by adding her expected payoff in the current step $t$ where the opponent is using a memory-limited count-based learning rule $f_m$ with a count history of $(c_1+c'_1,\cdots, c_i+c'_i+1,\cdots, c_{k_A}+c'_{k_A})$, plus $A$'s expected cumulative payoff in step $t-1$ where the history of her actions is $(c_1,\cdots,c_i+1,\cdots,c_{k_A}, c'_1,\cdots, c'_j-1,\cdots,c'_{k_A})$. That is:
\begin{equation}
\begin{aligned}
&V[c_1,\cdots,c_{k_A}][c'_1,\cdots,c'_{k_A}]\\
=& \max_{\substack{(i, j) \in [k_A]^2 \\ c'_j \geq 1}}\Big[V[c_1,\cdots,c_i+1,\cdots,c_{k_A}][c'_1,\cdots,c'_j-1,\cdots,c'_{k_A}]\\
+&\bar{u}^A(a_j, f_m(c_1+c'_1,\cdots,c_i+c'_i+1,\cdots,c_j+c'_j-1, \cdots c_{k_A}+c'_{k_A}))\Big]
\end{aligned}
\end{equation}

Assuming the value for $V[c_1,\cdots,c_i+1,\cdots,c_{k_A}][c'_1,\cdots,c'_j-1,\cdots,c'_{k_A}]$ was computed correctly, the above recurrence gives the optimal value for $V[c_1,\cdots,c_{k_A}][c'_1,\cdots,c'_{k_A}]$.

\smallskip
\noindent\textbf{Time Complexity.\ } 
Next, we analyze the time complexity of the algorithm.
The number of entries in the DP table is equal to the number of tuples $(c_1,\cdots,c_{k_A},c'_1,\cdots,c'_{k_A})$ whose components add up to $m$. The number of such tuples is 
\[\binom{m+2k_A-1}{2k_A-1} = O((m+2k_A-1)^{2k_A-1})\]

Computing each table entry takes $O(k_B\cdot {k_A}^3)$, where $O(k_A^2)$ is for taking the max over all pairs of $(i,j)$ in the recurrence and $O(k_Ak_B)$ is for computing the output of the memory-limited count-based learning rule $f_m$ in each time step. Therefore, the total time complexity is $O(k_B {k_A}^3(m+2k_A)^{2k_A})$.
\end{proof}

\begin{proof}[Proof of \Cref{thm:half-approx}]
First, we argue that~\Cref{alg:half-approx} achieves a $1/2$-approximation guarantee. Let $\OPT_{e}, \OPT_{o}$ denote $A$'s maximum expected cumulative payoff in even-numbered and odd-numbered blocks, respectively. Let $\OPT$ represent $A$'s maximum expected cumulative payoff over the entire time horizon. By definition, we have:
\[\OPT_e + \OPT_o \geq \OPT\]

Since the algorithm guarantees expected payoff of $\max(\OPT_e, \OPT_o)$ it ensures that player $A$ receives at least half of their optimal expected cumulative payoff, thus achieving a $1/2$-approximation.

The dominant component in the time complexity of the algorithm is the running time of the Dynamic Programming algorithm for computing $S_{2m}$ and $S_{m+p}$, which is $O(k_B{k_A}^3(m+2k_A)^{2k_A})$ by~\Cref{thm:DP-S-2m}.
\end{proof}

\section{Proof from \Cref{section:ptas}}
\label{appendix:proofs-from-ptas}

\begin{proof}[Proof of \Cref{thm:ptas}]
First, consider the case where $\lceil m/\eps \rceil$ divides $T$. In this case, in each timeblock, $A$ is maximizing its expected cumulative payoff in the last $\lceil m/\eps \rceil - m$ time steps. Let $\OPT_{\lceil m/\eps \rceil}$ denote the maximum expected payoff that $A$ can achieve on a block of size $\lceil m/\eps \rceil$ in isolation, i.e. when $T=\lceil m/\eps \rceil$. Then in the first block, $A$'s expected payoff is at least $\frac{\lceil m/\eps \rceil-m}{\lceil m/\eps \rceil} \OPT_{\lceil m/\eps \rceil} \geq (1-\eps) \OPT_{\lceil m/\eps \rceil}$. Now, in any subsequent block, the actions that $B$ is taking in the first $m$ time steps depends on the actions taken by $A$ in the previous block. However, after that, the actions that $B$ is taking only depend on the actions taken by $A$ in the current block. Since we permute over all possible action sequences for $A$ and take the one that maximizes $A$'s expected payoff in the last $\lceil m/\eps \rceil-m$ time steps, we can guarantee that $A$ is achieving at least $\frac{\lceil m/\eps \rceil-m}{\lceil m/\eps \rceil} \OPT_{\lceil m/\eps \rceil}$ in any subsequent block as well. As a result, $A$ is achieving at least $\frac{\lceil m/\eps \rceil-m}{\lceil m/\eps \rceil} \OPT \geq (1-\eps) \OPT$ in total.

Next, consider the case where $\lceil m/\eps \rceil$ does not divide $T$. Let $T \bmod (\lceil m/\eps \rceil) = p$. Suppose the time horizon is divided into blocks $B_1,\cdots, B_u$ where $|B_u| = p$ and $|B_i| = \lceil m/\eps \rceil$ for $1\leq i< u$. Here again, for each timeblock $B_i$ for $1\leq i\leq u$, $A$ maximizes its payoff after the first $m$ time steps. Similar to the argument in the previous case, in $B_1,\cdots, B_{u-1}$, A is getting at least $\frac{\lceil m/\eps \rceil-m}{\lceil m/\eps \rceil} \OPT_{\lceil m/\eps \rceil} \geq (1-\eps) \OPT_{\lceil m/\eps \rceil}$ payoff. In the last timeblock $B_u$, in the first $\min(m,p)$ time steps, we are not guaranteed to get any payoff. After that, we are getting the maximum possible payoff.
Therefore, in total, $A$ is guaranteed to receive at least $(1 - \eps)(\OPT - mH)$ payoff.

Finally, we analyze the time complexity of the algorithm. The algorithm enumerates $O(k_A^{\lceil m/\eps \rceil})$ possible sequences, and selects the one that maximizes $A$'s payoff after the first $m$ time steps. For each sequence, computing $A$'s payoff over this interval takes $O(k_Ak_B (\lceil m/\eps \rceil))$ time since the time complexity needed for computing the memory-limited count-based learning rule $f_m$ over each time step is $O(k_Ak_B)$. This results in a total time complexity of $O(k_A^{\lceil m/\eps \rceil + 1} k_B \lceil m/\eps \rceil)$.
\end{proof}

\section{Proofs from \Cref{section:additive_approx}}
\label{section:additive-approx-proofs}

We begin with a warm-up result for memory $m = 1$, and then establish the result for general memory $m \geq 1$.

\newcommand{\FTLone}{\mathsf{FTL}_1}
\subsection{Cyclic Strategy against Memory-1 Count-Based Learner in the Limit}
\label{section:ftl-1-optimal-cycle}

Here we show a result (\Cref{theorem:ftl-1-optimal-cycle}) for the maximal mean weight action cycle---sequence of actions that when repeated yields the maximal average utility per action among all cycles---against $f_1$. The running time of this algorithm is competitive with both the dynamic programming result from~\Cref{section:dp-algorithms} (drops dependence on $T$) and the approximation algorithms from~\Cref{section:approximation-algorithms}. Moreover, the approximation factor of this maximal mean weight cycle strategy converges to 1 as $T \to \infty$. All cycles produced by the algorithms are simple. If the overall optimal cycle of actions were not simple, then we can decompose it into simple cycles which need to be at least as good as the optimal simple cycle. This reduces the problem to finding only the optimal simple cycle.

\begin{theorem}
\label{theorem:ftl-1-optimal-cycle}
    Suppose that all payoffs in $u^A$ and $u^B$ are bounded in $[0, H]$. For $T \gg k_A^2$, repeating the cycle from~\Cref{algorithm:ftl-1-optimal-cycle}, which runs in $\bigo{k_A^3 k_B}$ time, gives expected cumulative payoff of at least $\OPT - (3k_A + 1) H$ against $f_1$, where $\OPT$ is the optimal expected cumulative payoff for a sequence of length $T$. Therefore, in the limit when $T \to \infty$ repeating the optimal cycle yields optimal expected average per-action utility against $f_1$.
\end{theorem}

\subsubsection{Weighted Transition Graph for Memory-1 Count-Based Learner}
\label{subsection:transition-graph-ftl-1}

\begin{definition}[Transition graph for Memory-1 Count-Based Learner]
\label{definition:ftl-1-transition-graph}
    Consider a weighted directed graph $G(V_1, E_1)$, so that $V_1 = \calA$, $E_1 = \{(a_i, a_j) \mid a_i, a_j \in \calA,\ i, j \in [k_A]\}$, and each directed edge $(a_i, a_j)$ has weight $w(a_i, a_j) = \bar{u}^A(a_j, b_i)$, where $b_i = f_1(a_i)$. Then $G_1(\calA, \calB, \bar{u}^A, \bar{u}^B) = G(V_1, E_1)$ is the \emph{transition graph} for the memory-1 count-based learning rule $f_1$ defined by the action spaces $\calA, \calB$ and the expected utility functions $\bar{u}^A$ and $\bar{u}^B$ for Player A and Player B, respectively. \Cref{algorithm:ftl-1-transition-graph} shows how to construct the transition graph given the action sets $\calA, \calB$ and the expected utility functions $\bar{u}^A$ and $\bar{u}^B$.
\end{definition}

\begin{algorithm}[htbp]
\caption{\textsc{TransitionGraphForFOne}: Building the transition graph $G_1$ for memory-1 count-based learner $f_1$}
\label{algorithm:ftl-1-transition-graph}

    \raggedright
    \setcounter{AlgoLine}{0}
    
    \KwIn{Action sets $\mathcal{A}, \mathcal{B}$, utility functions $\bar{u}^A, \bar{u}^B$, memory-1 count-based learner $f_1$}
    \KwOut{The transition graph $G_1$ for $f_1$}
    
    $V_1 \leftarrow \mathcal{A}$\;
    
    $E_1 \leftarrow \emptyset$\;
    
    \For{$a_i \in V_1$}{
        $b_i \leftarrow f_1(a_i)$\;
        
        \For{$a_j \in V_1$}{
            $w \leftarrow \bar{u}^A(a_j, b_i)$\;
            
            Add edge $(a_i, a_j)$ with weight $w$ to $E_1$\;
        }
    }
    
    $G_1 \leftarrow G(V_1, E_1)$ \tcp{Construct graph with vertices $V_1$ and edges $E_1$}
    \Return{$G_1$}\;
\end{algorithm}

\begin{proposition}[Reduction to Maximal Weight Walk in $G_1$]
\label{proposition:ftl-1-reduction}
    Finding a sequence of actions that maximizes the expected cumulative payoff for the Optimizer (Player A) over a time horizon $T$ against a memory-1 count-based learner $f_1$ (Player B) is equivalent to finding the maximum weight walk\footnote{Not necessarily simple, so vertices are allowed to repeat.} of length $T - 1$ in the corresponding transition graph $G_1 \coloneq G_1(\calA, \calB, \bar{u}^A, \bar{u}^B)$.
\end{proposition}

\begin{proof}
    The Optimizer's action at time $t - 1$ is the only piece of information required to determine the action of $f_1$ at time $t$. Therefore, we can model the problem using the transition graph $G_1$ (see~\Cref{definition:ftl-1-transition-graph}). Recall that the utility gained by the Optimizer at time $t$ is $\bar{u}^A(a_t, b_t)$. This is exactly the weight of the edge $(a_{t - 1}, a_t)$, namely $w(a_{t - 1}, a_t) = \bar{u}^A(a_{t}, b_t)$, where $b_t = f_1(a_{t - 1})$.
    
    Observe that an action sequence $(a_1, a_2, \dots, a_T)$ for the Optimizer corresponds to a walk of length $T - 1$ in $G_1$. Then the total utility accumulated by the Optimizer is the sum of the weights of the edges in this walk, plus the utility from the initial action $a_0$, which is fixed by assumption. Therefore, to find the action sequence $(a_1, a_2, \dots, a_T)$ that maximizes the expected utility it is enough to find the walk of length $T - 1$ in $G_1$ with maximal total weight.

    Notice that from above the reverse is also true: the action sequence of length~$T$ of maximal utility yields a walk of length $T - 1$ of maximal total weight.
\end{proof}

\subsubsection{Optimal Cycle Against a Memory-1 Count-Based Learner}
\label{subsection:ftl-1-optimal-cycle}

\begin{algorithm}[htbp]
\caption{Optimal Cycle Against Memory-1 Count-Based Learning Rule $f_1$}
\label{algorithm:ftl-1-optimal-cycle}
    
    \raggedright
    \setcounter{AlgoLine}{0}
    
    \KwIn{Action sets $\mathcal{A}, \mathcal{B}$, utility functions $\bar{u}^A, \bar{u}^B$, memory-1 count-based learning rule $f_1$}
    \KwOut{Cycle $C^*$ representing the optimal cyclic strategy sequence}
    
    $G_1 \leftarrow \textsc{TransitionGraphForFOne}(\mathcal{A}, \mathcal{B}, \bar{u}^A, \bar{u}^B)$ \tcp{\Cref{algorithm:ftl-1-transition-graph}}
    
    $C^* \leftarrow \textsc{KarpMaxMeanCycle}(G_1)$ \tcp{Find maximum mean weight cycle, \Cref{section:karp-max-mean-weight-cycle}}
    
    \Return{$C^*$}\;
\end{algorithm}

\begin{proposition}
\label{proposition:ftl-1-optimal-running-time}
    \Cref{algorithm:ftl-1-optimal-cycle} finds the cycle of actions yielding the maximal average utility per time step in time $\bigo{k_A^3 k_B}$, where $k_A = \abs{\mathcal{A}}$ and $k_B =~\abs{\mathcal{B}}$.
\end{proposition}
\begin{proof}
    The correctness of~\Cref{algorithm:ftl-1-optimal-cycle} follows directly from~\Cref{proposition:ftl-1-reduction} and the correctness of Karp's Maximal Mean Weight Cycle algorithm. Next, we calculate the time complexity for each step separately.
    \begin{itemize}[noitemsep, topsep=0pt, parsep=0pt, partopsep=0pt]
        \item \textbf{Constructing the transition graph:}\ The graph $G_1$ has $k_A$ vertices. To construct the edges, the outer loop runs $k_A$ times and the inner loop also runs $k_A$ times, creating $k_A^2$ edges in total. Computing the opponent's response with the count-based learning rule $f_1$ takes $\bigo{k_A k_B}$ time and is performed once per edge. Therefore, the total time for constructing the transition graph is $\bigo{k_A^3 k_B}$.
        \item \textbf{Finding the Maximum Mean Weight Cycle:}\ Karp's algorithm~\cite{karp1978characterization}, with its correction~\cite{chaturvedi2017note}, on a graph $G(V, E)$ runs in $\bigo{\abs{V} \cdot \abs{E}}$ time. Therefore, the running time for $G_1$ is $\bigo{k_A \cdot k_A^2} = \bigo{k_A^3}$.
    \end{itemize}
\end{proof}

\subsubsection{Proof of \Cref{theorem:ftl-1-optimal-cycle}}

\begin{proof}
    All cycles produced by the algorithm are simple and have size at most $k_A$. Notice that if the overall optimal cycle of actions was not simple, then we can decompose it into simple cycles which need to be at least as good as the optimal simple cycle. Therefore, it is enough to consider only the simple cycles produced by Karp's algorithm.

    Consider the optimal sequence of actions for time horizon $T$. Next, consider any action $a$ that has been repeated at least three times. Substituting the segments between consecutive executions of $a$ with the maximum-mean-weight (simple) cycle that contains $a$ would yield a better per-action utility compared to the per-action utility between consecutive executions of $a$ in the original sequence. However, this can decrease the length of the sequence, in this case we can repeat the cycle that contains $a$ enough times to cover a sequence with length in $[T - k_A, T]$ (since any simple cycle has size at most $k_A$). Therefore, we would lose the per-action expected payoff from at most $k_A$ actions at the end.

    Now observe that any action that has been repeated at least three times must be first used within the first $2k_A + 1$ actions (using each of the other actions twice and that one action thrice). Therefore, because the above procedure is not optimizing for the beginning of the action sequence it would lose at most the per-action expected payoff of at most $2k_A + 1$ actions at the beginning of the sequence. Combining the two facts, the expected cumulative payoff would be at least $\OPT - (3k_A + 1) H$, where $\OPT$ is the optimal expected cumulative payoff for a sequence of length $T$.
\end{proof}

\subsection{Cyclic Strategy against Memory-Limited Count-Based Learner in the Limit}
\label{section:ftl-m-optimal-cycle}

We can extend the results from \Cref{section:ftl-1-optimal-cycle} to the memory-limited count-based learning rule $f_m$ with any finite memory $m \geq 1$. The running time roughly matches that of the Dynamic Programming algorithms from~\Cref{section:dp-algorithms} up to dependence on $T$.

\begin{theorem}
\label{theorem:ftl-m-optimal-cycle}
    Suppose that all payoffs in $u^A$ and $u^B$ are bounded in $[0, H]$. For $T \gg k_A^{2m}$, repeating the cycle from~\Cref{algorithm:ftl-m-optimal}, which runs in $\bigo{k_A^{2m + 1} + k_A^{m} (m + k_A k_B)}$ time, gives an expected cumulative payoff of at least $\OPT - (3k_A^m + m) H$ against $f_m$, where $\OPT$ is the optimal expected cumulative payoff for a sequence of length~$T$. Therefore, in the limit when $T \to \infty$, repeating the optimal cycle yields optimal expected average per-action utility against $f_m$.
\end{theorem}

\begin{algorithm}[tbp]
\caption{\textsc{TransitionGraphForFM}: Building the transition graph $G_m$ for memory-limited count-based learning rule $f_m$}
\label{algorithm:ftl-m-transition-graph}

    \raggedright
    \setcounter{AlgoLine}{0}
    
    \KwIn{Action sets $\mathcal{A}, \mathcal{B}$, utility functions $\bar{u}^A, \bar{u}^B$, memory length $m$, memory-limited count-based learning rule $f_m$}
    \KwOut{The transition graph $G_m$ for memory-limited count-based learning rule $f_m$}
    
    $V_m \leftarrow \mathcal{A}^m$ \tcp{Set of all action sequences of length $m$}
    $E_m \leftarrow \emptyset$\;
    
    \For{$s = (a_1, \dots, a_m) \in V_m$}{
        $b_s \leftarrow f_m(s)$\;
        
        \For{$a_{\text{new}} \in \mathcal{A}$}{
            $s' \leftarrow (a_2, \dots, a_m, a_{\text{new}})$\;
            
            $w \leftarrow \bar{u}^A(a_{\text{new}}, b_s)$\;
            
            Add edge $(s, s')$ with weight $w$ to $E_m$\;
        }
    }
    
    $G_m \leftarrow G(V_m, E_m)$ \tcp{Construct graph with vertices $V_m$ and edges $E_m$}
    \Return{$G_m$}\;
\end{algorithm}

\subsubsection{Weighted Transition Graph for Memory-Limited Count-Based Learner}
\label{subsection:transition-graph-ftl-m}

\begin{definition}[Transition graph for Memory-Limited Count-Based Learner]
\label{definition:ftl-m-transition-graph}
    Let $G_m(\bar{u}^A, \bar{u}^B) = G(V_m, E_m)$, where $V_m = \calA^m$ and $E_m = \{(s, s') \mid s, s' \in V_m, s[1:] = s'[:-1]\}$,\footnote{In Python notation $s[1:]$, where indexing starts from zero, denotes the last $m - 1$ coordinates of $s$, and $s'[:-1]$ denotes the first $m - 1$ coordinates of $s'$.} be a weighted directed graph. Notice that each vertex is a sequence of $m$ actions, and a vertex connects two action sequences if the end of the first one $s[1:]$ matches the beginning of the other. Each directed edge $(s, s')$ has weight $w(s, s') = \bar{u}^A(s'[m], b_s)$, where $b_s = f_m(s)$. Then $G_m(\calA, \calB, \bar{u}^A, \bar{u}^B) = G(V_m, E_m)$ is the \emph{transition graph} for $f_m$ defined by the action spaces $\calA, \calB$ and the utility functions $\bar{u}^A$ and $\bar{u}^B$ for Player A and Player B, respectively. Constructing the transition graph for memory-limited count-based learning $f_m$ (see \Cref{algorithm:ftl-m-transition-graph}) follows an analogous procedure to the one for $f_1$ from \Cref{algorithm:ftl-1-transition-graph}.
\end{definition}

\begin{proposition}[Reduction to Maximal Weight Walk]
\label{proposition:ftl-m-reduction}
    Fix the first $m$ actions played by $A$, then finding a sequence of actions that maximizes the expected cumulative payoff for the Optimizer (Player A) over a time horizon $T + m$ against a memory-limited count-based learning rule $f_m$ (Player B) is equivalent to finding the maximum weight walk\footnote{Not necessarily simple, so vertices are allowed to repeat.} of length $T - 1$ in the corresponding transition graph $G_m \coloneq G_m(\calA, \calB, \bar{u}^A, \bar{u}^B)$.
\end{proposition}
\begin{proof}
    The proof is analogous to the proof of~\Cref{proposition:ftl-1-reduction} after we fix the first $m$ actions.
\end{proof}

\subsubsection{Optimal Cycle against a Memory-Limited Count-Based Learner}
\label{subsection:ftl-m-optimal-cycle}

\begin{algorithm}[htbp]
\caption{Optimal Cycle against Memory-Limited Count-Based Learner}
\label{algorithm:ftl-m-optimal}

    \raggedright
    \setcounter{AlgoLine}{0}
    
    \KwIn{Action sets $\mathcal{A}, \mathcal{B}$; utility functions $\bar{u}^A, \bar{u}^B$; memory length $m$}
    \KwOut{Cycle $C^*$ representing the optimal long-term strategy}
    
    $G_m \leftarrow \textsc{TransitionGraphForFM}(\mathcal{A}, \mathcal{B}, \bar{u}^A, \bar{u}^B, m)$ \tcp{\Cref{algorithm:ftl-m-transition-graph}}
    
    $C^* \leftarrow \textsc{KarpMaxMeanCycle}(G_m)$ \tcp{Find maximum mean weight cycle in $G_m$}
    
    \Return{$C^*$}\;
\end{algorithm}

\begin{proposition}
\label{proposition:ftl-m-optimal-running-time}
    \Cref{algorithm:ftl-m-optimal} finds the cycle of actions that yields the maximal expected average utility per time step in time $O(k_A^{2m + 1} + k_A^{m} (m + k_A k_B))$.
\end{proposition}
\begin{proof}
    We calculate the time complexity for each procedure.
    \begin{itemize}[noitemsep, topsep=0pt, parsep=0pt, partopsep=0pt]
        \item \textbf{Constructing the transition graph:} The graph $G_m$ has $|V_m| = k_A^m$ vertices. From each vertex, there are $k_A$ outgoing edges, for a total of $|E_m| = k_A^m \cdot k_A = k_A^{m + 1}$ edges. The opponent's response must be calculated for each of the $k_A^m$ source vertices, and takes $\bigo{m}$ time to compute the count vector and $\bigo{k_A k_B}$ to compute the count-based learning rule $f_m$.
        Therefore, the running time for constructing the transition graph is dominated by $O(k_A^{m} (m + k_A k_B))$.
       
        \item \textbf{Finding the Maximum Mean Weight Cycle:} Karp's algorithm~\cite{karp1978characterization}, with its correction~\cite{chaturvedi2017note}, on $G_m$ runs in $O(\abs{V_m} \cdot \abs{E_m}) = O(k_A^m \cdot k_A^{m + 1}) = O(k_A^{2m + 1})$ time.
    \end{itemize}
\end{proof}

\subsubsection{Proof of \Cref{theorem:ftl-m-optimal-cycle}}

\begin{proof}
    The proof follows analogously to that of \Cref{theorem:ftl-1-optimal-cycle}. By \Cref{definition:ftl-m-transition-graph}, the transition graph $G_m$ has $\abs{V_m} = k_A^m$ vertices. All cycles produced by the algorithm are simple and thus contain at most $k_A^m$ edges. If the overall optimal walk in $G_m$ is not a simple cycle, it can be decomposed into simple cycles which must perform at least as well as the optimal simple cycle.

    Consider the optimal sequence of actions for time horizon $T$, which corresponds to a maximum weight walk in $G_m$ after fixing the initial $m$ actions. Next, consider any state $s \in V_m$ (a sequence of $m$ actions) that has been visited at least three times. Substituting the path segments between consecutive visits to $s$ with the maximum-mean-weight simple cycle containing $s$ yields a better or equal per-action utility compared to the original sequence. Because this substitution may decrease the total sequence length, we can repeat the optimal cycle containing $s$ enough times to cover a sequence of length in $[T - k_A^m, T]$ (since any simple cycle has length at most $k_A^m$). Thus, we might lose the expected payoff from at most $k_A^m$ actions at the end of the sequence.

    By the Pigeonhole Principle, any state visited at least three times must first achieve this within the first $2k_A^m + 1$ edges of the walk (visiting all other $k_A^m - 1$ states twice and $s$ thrice). Accounting for the $m$ initial actions required to define the first vertex in $G_m$, the unoptimized prefix of the sequence costs the per-action expected payoff of at most $2k_A^m + m$ actions. Combining these two boundary losses, the expected cumulative payoff is bounded below by $\OPT - (3k_A^m + m) H$, where $\OPT$ is the optimal expected cumulative payoff for a sequence of length $T$.
\end{proof}

\subsection{Karp's Algorithm for Maximum Mean Weight Cycle}\label{section:karp-max-mean-weight-cycle}

\begin{definition}[Maximum Mean Weight Cycle]
    Given a strongly connected directed graph $G = (V, E)$ and a weight function $w: E \to \mathbb{R}$, a simple cycle is a sequence of vertices $(v_1, v_2, \dots, v_l, v_1)$ where all $v_i$ are distinct. The \textbf{mean weight} of a cycle is:
    \[ 
        \mu(C) = \frac{1}{l} \sum_{i = 1}^{l} w(v_i, v_{i + 1})
    \]
    where $v_{l + 1} = v_1$. The \textbf{maximum/minimum mean weight cycle} is the simple cycle $C^*$ in $G$ for which $\mu(C^*)$ is maximized/minimized.
\end{definition}

\begin{theorem}[Karp's Theorem, Maximization Form~\cite{karp1978characterization}]
    Let $G = (V, E)$ be a directed graph. For any arbitrary source vertex $s \in V$, let $D_l(v)$ be the maximum weight of a path of exactly $l$ edges from $s$ to $v$. If the graph contains a cycle with positive mean weight, then the maximum mean weight $\lambda$ is given by:
    \[
        \lambda = \max_{v \in V} \left( \min_{0 \le l < \abs{V}} \frac{D_{\abs{V}}(v) - D_l(v)}{\abs{V} - l} \right).
    \]
\end{theorem}

\begin{proposition}
    Karp's algorithm~\cite{karp1978characterization} computes the maximum mean weight of a cycle in a graph $G(V, E)$ in $\bigo{\abs{V} \cdot \abs{E}}$ time. The algorithm can be further modified to also compute a cycle that achieves the maximal mean weight~\cite{karp1978characterization,chaturvedi2017note}.
\end{proposition}

The algorithm is originally presented for finding the value of the minimum mean weight cycle, however, it can be modified to also construct the maximum mean weight cycle~\cite{karp1978characterization,chaturvedi2017note}. Later \cite{dasdan2002faster} developed a faster algorithm specifically for the problem of finding the maximum mean weight cycle.

\section{Proof from \Cref{sec:unknown-opponent}} \label{appendix:missing-proof-unknown-opponent}
\begin{proof}[Proof of \Cref{thm:unknown-opponent}]
 \textit{Reduction to stochastic bandits:} For each learning rule $f \in \mathcal{F}$, we compute the sequence of length $2m$ which has the highest expected payoff against $f$ over the last $m$ steps. Call the sequence for learning rule $f$ $S_{f2m}$, and denote the set of all such sequences $\mathbb{S}$. Divide the time horizon $T$ into blocks of $2m$ steps each, and calling these blocks $t'_1, t'_2, \dots t_{T'}$, where $T' = \lfloor \frac{T}{2m}\rfloor$; call any remaining time steps at the end block $R$. We will treat each block $t'_i$ as a single time step in the stochastic bandits context, in which pulling arm $i$ will correspond to playing $S_{f_i2m}$, and the realized reward of pulling arm $i$ is the realized payoff during the last $m$ steps of $S_{f_i2m}$. 

 Using MOSS as $\mathsf{Alg_{SB}}$, we can apply the $O(\sqrt{KT})$ regret bound from \cite{stochastic_bandits}. In our case $K=|\mathcal{F}|$ and $T=T'$; the bound relies on payoffs in $[0,1]$, so we apply it to the payoff over the last $m$ rounds divided by $m$, and rescale the resulting regret by $m$. Substituting $|\mathcal{F}|$ and $T' = \lfloor \frac{T}{2m}\rfloor$, then, we get a regret bound of $O(\sqrt{|\mathcal{F}|mT})$ compared to the payoff we would get playing the best arm over the entire sequence of $T'$ time steps.

It remains to show that playing the ``best arm'' (i.e., the sequence $S_{f2m}$ with the highest expected payoff over the last $m$ steps---call this the ``best sequence'') over all $T'$ time steps, followed by arbitrary actions over the remaining $T-2mT'$ time steps, obtains at least $\frac{1}2{\OPT} - \frac{2m -1}{2}$ in expectation over the original $T$ time steps. 

Let $\OPT_{2mT'}$ be the expected payoff of the optimal sequence over the first $T'*2m$ time steps, and $\OPT_{R}$ be the expected payoff of the optimal sequence over the remaining $T-T'*2m$ time steps. Then:

\begin{align*}
 \OPT_{2mT'} + \OPT_{R} &\geq \OPT  \\
\frac{1}{2}\OPT_{2mT'} + \frac{1}{2}\OPT_{R}  &\geq \frac{1}{2}\OPT  \\ 
\frac{1}{2}\OPT_{2mT'} &\geq \frac{1}{2}\OPT - \frac{1}{2}\OPT_{R} \\ 
\end{align*}

There are at most $2m-1$ remaining time steps and all payoffs are at most 1, so $\OPT_R \leq 2m-1$. So we have: 

\begin{align*}
  \frac{1}{2}\OPT_{2mT'} &\geq \frac{1}{2}\OPT - \frac{2m -1}{2} \\
\end{align*}

It remains to show that repeating the best sequence of length $2m$ over the first $2mT'$ time steps, then playing arbitrary actions for the remainder of the game, achieves at least $\frac{1}{2}\OPT_{2mT'}$ in expectation. Since all payoffs are non-negative, it suffices to show $\frac{1}{2}\OPT_{2mT'}$ is gained in expectation during the first $2mT'$ time steps. 

During the first $2mT'$ time steps, we repeat the sequence $S_{f2m}$ which has the highest expected payoff over the last $m$ steps among all sequences in $\mathbb{S}$. Recall that $\mathbb{S}$ included $S_{f^*2m}$, where $f^*$ is the true opponent, so the expected payoff of $S_{f2m}$ over the last $m$ steps is at least that of $S_{f^*2m}$ over the last $m$ steps. 

Now, we claim that the highest expected score possible over any block of $m$ time steps is achieved in the second $m$ steps of $S_f2m$. This is immediate for any block of $m$ time steps beginning after time step $t > m$ (in which the entire memory is filled) as $S_{f2m}$ chooses the full memory in the first $m$ steps such that the expected payoff in the second $m$ steps is maximized. For any block of $m$ steps beginning at time step $t \leq m$ (for which there are fewer than $m$ actions in the history) it also holds as this is equivalent to the special case with history length $m$ where the first $m - t + 1$ actions are blank actions. 

Dividing the first $2mT'$ time steps into blocks of length $m$, the optimal strategy must in expectation achieve at least half of its expected utility in either the odd or even blocks. By playing $S_{f2m}$ repeatedly, we maximize the score in the even blocks. Since there are the same number of odd and even blocks in $2mT'$, by what we just showed above, this must be at least as good as maximizing the expected payoff in the odd blocks. Therefore, we achieve at least $\frac{1}{2}\OPT_{2mT'}$ in expectation.
\end{proof}

\section{Proofs from Section \ref{section:hardness-results}}
\label{appendix:missing-proofs-hardness}

\begin{proof}[Proof of Theorem \ref{thm:np_hardnes} (Part 1)] \textsc{It is NP-hard to compute a sequence of actions which attains payoff within any constant multiplicative factor of the optimal payoff against \FTL, even in zero-sum games.}
    
Suppose we have a polynomial-time algorithm $\mathbb{A}$ for finding the optimal sequence against $\FTL$. We will show that we could use it to solve 3-SAT in polynomial time using $\mathbb{A}$, which is not possible unless $P=NP$. 

\begin{figure}[htbp]
\centering
\caption{Utility function $u^B(a,b)$}
\[
u^B(a,b)=
\begin{cases}
0 & b \text { is } X_i \text{ and } a \in \{X_{iT}, X_{iF}\}\\
-n + 1 & b \text { is } X_i \text{ and } a \in \{X_{jT}, X_{jF}\} \text{, } j \neq i\\
-2n + 2 & b \text { is } C_j\text{, $a$ is } X_{iT} \text { and } x_i \in c_{j}\\
-2n +2 & b \text { is } C_j\text{, $a$ is } X_{iF} \text { and } \neg x_i \in c_{j}\\
-n + 3 & b \text { is } C_j\text{, $a$ is } X_{iT} \text { and } x_i \notin c_{j}\\
-n + 3 & b \text { is } C_j\text{, $a$ is } X_{iF} \text { and } \neg x_i \notin c_{j}\\ 
0 & b \text{ is not $\Done$ and } a \text{ is $\Done$}\\
-n + 2 & b \text{ is $\Done$ and } a \text{ is not $\Done$} \\
-L & b \text{ is $\Done$ and } a \text{ is $\Done$} \\
\end{cases}
\]
\label{eq:utility_function}
\end{figure}

Given an instance of 3-SAT with $m$ clauses $c_1, ..., c_m$ and $n$ variables $x_1, ..., x_n$ (we assume $n \geq 3$ for simplicity), we set up a game $G$ as follows. Let player A have actions $X_{iT}$ and $X_{iF}$ for each variable $x_i$ in the formula, plus one additional action, $\Done$, such that $\mathcal{A} = \{X_{1T}, X_{1F}, ..., X_{nT}, X_{nF}, \Done \}$. Let player B have an action $X_{i}$ for each variable $x_i$ in the formula, and action $C_j$ for each clause $c_j$ in the formula, plus one additional action, $\Done$, such that $\mathcal{B} = \{X_{1}, ..., X_{n}, C_1, ..., C_m, \Done \}$. Let $G$ be zero-sum and define Player B's payoffs according to~\Cref{eq:utility_function}; let $L$ be arbitrarily large and as defined within the proof. This will produce a game matrix like the example shown in~\Cref{fig:hardness_example_payoffs}.

Let $T=n+1$, and assume that the opponent breaks ties among actions deterministically according to the order $X_1, ... X_n, C_1, ..., C_m, \Done$ (from left to right as shown in the example in \Cref{fig:hardness_example_payoffs}). 

\begin{figure}[htbp]
\centering
\caption{Opponent's payoffs given example clauses $c_1=v_1 \vee \neg v_2 \vee v_n$, $c_m= \neg v_1 \vee v_2 \vee v_n$}
\renewcommand{\arraystretch}{1.2}
\begin{tabular}{l|*{3}{c}|*{3}{c}|c}
       & \multicolumn{3}{c|}{\textbf{Assignments}} & \multicolumn{3}{c|}{\textbf{Clauses}} & \textbf{} \\
       & $X_1$ & ... & $X_n$ & $C_1$ & ... & $C_m$ & $\Done$ \\
\specialrule{1pt}{1pt}{1pt} % Thicker rule
$X_{1T}$ &    $0$   &       &     $-n + 1$   &  $-2n+2$    &       &    $-n+3$   &    $-n+2$   \\
$X_{1F}$ &   $0$    &       &     $-n + 1$  &   $-n+3$    &       &   $-2n+2$    &   $-n+2$    \\
$X_{2T}$ &     $-n + 1$  &       &     $-n + 1$   &     $-n+3$ &       &   $-2n+2$    &    $-n+2$  \\
$X_{2F}$ &    $-n + 1$   &      &       $-n + 1$ &     $-2n+2$  &       &    $-n+3$   &   $-n+2$    \\
$\vdots$ & $\vdots$ & $\vdots$ & $\vdots$ & $\vdots$ & $\vdots$ & $\vdots$ & $\vdots$ \\
$X_{nT}$ &    $-n + 1$    &       &   $0$    &   $-2n+2$    &       &   $-2n+2$   &    $-n+2$   \\
$X_{nF}$ &   $-n + 1$     &       &    $0$   &    $-n+3$   &       &   $-n+3$    &     $-n+2$  \\
$\Done$ &    0   &   ...    &    0   &    0   &   ...    &   0    &  $-L$  \\
\end{tabular}
\label{fig:hardness_example_payoffs}
\end{figure}

We'll show that when $L$ is set large enough, the optimizer can achieve an $\alpha$-approximately optimal payoff in this game if and only if they play actions which correspond to a satisfying assignment, followed by $\Done$ (for constant $\alpha$ > 0). When $L$ is chosen large enough, if it is possible to get the payoff $L$ corresponding to $(\Done, \Done)$, the optimizer must do so to achieve $\alpha$-approximation of the the optimal payoff. The maximum payoff the optimizer could achieve without $L$ over all $T=n+1$ rounds is at most $(n+1)(2n-2) = 2n^2 - 2$, so since all payoffs for the optimizer are nonnegative, $L>\frac{2n^2 -2}{\alpha}$ suffices. Since $\mathbb{A}$ finds a sequence which maximizes the optimizer's payoff, it follows that there exists a satisfying assignment if and only if the sequence $S = \mathbb{A}(G)$ attains payoff $\geq L$; it's easy to check this in $O(Tk_B)$ time. Moreover, if there is a satisfying assignment, it corresponds to the first $n$ actions of $S$ ($X_{iT}$ corresponds to setting $x_i=\True$, and $X_{iF}$ corresponds to setting $x_i=\False$). 

First, we'll show that if a satisfying assignment exists, the optimizer can get payoff $\geq L$ by playing $n$ actions corresponding to a satisfying assignment, followed by $\Done$, to collect the payoff of $L$ in round $n+1$. It suffices to show that the optimizer gets payoff $L$ in the last round, as all the optimizer's payoffs are nonnegative. If the optimizer plays a valid assignment which satisfies all clauses, each ``assignment'' column will have score $(n-1)(-n +1) + 1(0) = -n^2 + 2n -1$ for the opponent (this corresponds to each variable being assigned exactly once), and each ``clause'' column will have score $\leq (n-1)(-n+3) + 1(-2n+2) = -n^2 + 2n - 1$ (corresponding to each clause containing at least one literal which evaluates to true). The $\Done$ column will have the highest score, $n(-n + 2) = -n^2 + 2n$, so the opponent will play $\Done$. Therefore, the optimizer can also play $\Done$ to get payoff $L$ in that round.

Next, we show that if the optimizer does not play a satisfying assignment followed by $\Done$, they cannot get payoff $\geq L$. As explained above, for appropriate $L$, the optimizer can only achieve this maximum payoff when playing $\Done$, so suppose the optimizer plays $\Done$ when they have not played a satisfying assignment. We will show that in that case, the opponent would not play $\Done$, and therefore the optimizer cannot get payoff $L$. Once the optimizer has played $\Done$ once, the opponent will not play $\Done$ during any future round, as the large negative score $-L$ incurred for appropriately large $L$ $(> 2n^2 - 2$) is worse than the worst possible score of any other column (because $T=n+1$, the worst possible negative score of any other column is $-2n^2 + 2$). It follows that it is only possible to get the payoff $L$ the first time the optimizer plays $\Done$, so we can consider only that case below.

Suppose the optimizer plays $\Done$ for the first time in round 1. By our tie-breaking assumption, the opponent will not play $\Done$ in round 1, so the optimizer will not get payoff $\geq L$.

Suppose the optimizer plays $\Done$ for the first time in some other round $t$, where $1 < t \leq  n$. When the opponent is choosing their action for round $t$, the score for the $\Done$ column must be $(t-1)(-n + 2) = -tn + 2t + n - 2$. Meanwhile, since $\Done$ was not played before round $t$ and $t > 1$, at least one variable assignment action must have been played, so the score of some assignment column must be at least $(t-2)(-n+1) = -tn + t + 2n -2$. Since $t \leq n$, the score of the $\Done$ column cannot be strictly greater than the score of every assignment column, so by our tie-breaking assumption, the opponent will not play $\Done$ and the optimizer will not get payoff $\geq L$.

Finally, suppose the optimizer plays $\Done$ for the first time in round $T=n+1$, but their actions over the previous $n$ rounds do not correspond to a satisfying assignment. There are two ways this could happen: either the assignment is not valid or a clause is not satisfied. Consider the case where the assignment is not valid. Since the optimizer did not play $\Done$ in the first $n$ rounds, they must have played $n$ assignment actions, so some variable must be assigned more than once. For any variable which is assigned more than once, its corresponding column will have score $\geq (n-2)(-n+1) = -n^2 +3n -2$ for the opponent, which is greater than the score of the $\Done$ column at $n(-n + 2) = -n^2 +2n$ for $n \geq 3$, so the opponent will not play $\Done$ in round $n+1$ and the optimizer will not collect payoff $L$. If the assignment is valid but any clause is not satisfied, the score for the column corresponding to the unsatisfied clause is $n(-n+3) = -n^2 + 3n$ which is greater than the score of the $\Done$ column ($-n^2 + 2n$), so again the opponent will not play $\Done$ and the optimizer cannot collect payoff $L$. 
\end{proof}

\begin{proof}[Proof of Theorem \ref{thm:np_hardnes} (Part 2)] \textsc{It is NP-hard to compute a sequence of actions which attains payoff within any $T^\alpha$ additive error (for constant $0 < \alpha < 1$) of the optimal payoff against \FTL, even in zero-sum games.}

To show this, we can largely recycle the proof of \ref{thm:np_hardnes} (Part 1). Given a 3-SAT formula $\mathcal{F}$, we set up a game matrix $G'$ composed of an arbitrarily large number ($c$) of copies of the game $G$ (as defined in the proof above based on $\mathcal{F}$) along its diagonal, with all other payoffs set to 0. Let $T = c(n+1)$. We will refer to each copy of $G$ as a ``block''. Similarly to in the previous proof, $L$ must be set appropriately large; $L > T(2n -2) + T^{\alpha}$ will suffice here.\footnote{The reasoning for this is similar to in the previous proof: $T(2n-2)$ is the maximum cumulative payoff achievable without $L$, so the $T(2n -2)$ term alone in $L > T(2n -2) + T^{\alpha}$ makes obtaining $L$ necessary for obtaining an optimal payoff. The added $T^\alpha$ rules out an additive $T^\alpha$ approximation. $L > T(2n -2) + T^{\alpha}$ also ensures opponent will not play $\Done_b$ if the optimizer has played $\Done_b$ in the past (the maximum negative score of any column besides $\Done_b$ is $T(-2n + 2)$, and the $-L$ score for the $\Done_b$ column in this case will be strictly worse).}

The core idea is that maximizing the optimizer's expected cumulative payoff will require the optimizer to get as many of the $L$ payoffs corresponding to entries $(\Done_b, \Done_b)$ for each block $b$ as possible. In fact, we'll show it will require the optimizer to get all of them. 

It follows from the proof of Part 1 that the optimizer can get the payoff related to $(\Done_b, \Done_b)$ for block $b$ if and only if they play $n$ of that block's actions corresponding to a satisfying assignment of $\mathcal{F}$, followed by $\Done_b$. It also follows that each can only be attained once. To see that the reasoning above still holds here, we will note two things. First, since the blocks are along the diagonal and all other payoffs are 0, only the optimizer's actions corresponding to that block (i.e., the rows of that block) have any impact on the score of columns of that block. We can therefore apply the reasoning above to each block independently, considering only the rounds in which the optimizer played a row corresponding to the block in question. Second, since $T=c(n+1)$ rounds and there are $c$ blocks, the optimizer has exactly the number of rounds necessary to get the maximum payoff for each block: this implies that if there exists a satisfying assignment to $\mathcal{F}$, the optimizer must never play more than $n+1$ actions corresponding to a single block. 

Similarly to in the proof of Part 1, we can check if there is a satisfying assignment by seeing if the payoff $L$ for $(\Done_b, \Done_b)$ for any block $b$ is achieved (we can check this in $O(Tk_B)$ time); if so, we can recover a satisfying assignment to $\mathcal{F}$ from the other actions played corresponding to block $b$ as above ($X_{iTb}$ corresponds to $x_i = \True$, and $X_{iFb}$ corresponds to $x_i=\False$). To see where the $T^\alpha$ additive bound comes from, we can set $c=\ceil{T^\alpha}$.
\end{proof}

\end{document}